\def\draft{0}  

\ifnum\draft=2
\documentclass[10pt, conference, compsocconf]{IEEEtran}
\usepackage[cmex10]{amsmath}
\usepackage{amsfonts}
\else
\documentclass[11pt]{article}
\usepackage[letterpaper,hmargin=1.25in,vmargin=1.25in]{geometry}
\usepackage{amsfonts,url,ifthen,epsfig,amsmath,amssymb}
\usepackage{hyperref}
\fi

\ifnum\draft=1
\newcommand{\Rnote}[1]{{\bf [Ronen's Note: #1]}}
\newcommand{\Anote}[1]{{\bf [Alon's Note: #1]}}
\newcommand{\Nnote}[1]{{\bf [Noam's Note: #1]}}
\else
\newcommand{\Rnote}[1]{}
\newcommand{\Anote}[1]{}
\newcommand{\Gnote}[1]{}
\newcommand{\Nnote}[1]{}
\fi

\newcommand{\remove}[1]{}

\newtheorem{theorem}{\bf Theorem}[section]
\newtheorem{definition}[theorem]{\bf Definition}
\newenvironment{procthm}{\begin{theorem}}{\end{theorem}}
\newenvironment{procdef}[1]{\begin{definition}[#1]}{\end{definition}}
\newtheorem{lemma}[theorem]{\bf Lemma}
\newtheorem{proposition}[theorem]{\bf Proposition}
\newtheorem{corollary}[theorem]{\bf Corollary}
\newtheorem{claim}[theorem]{\bf Claim}

\newtheorem{remk}[theorem]{\bf Remark}

\newtheorem{examp}[theorem]{\bf Example}
\newenvironment{remark}{\begin{remk} \begin{normalfont}}{\end{normalfont}
\end{remk}}
\newenvironment{ex}{\begin{examp} \begin{normalfont}}{\end{normalfont}
\end{examp}}

\def\FullBox{\hbox{\vrule width 8pt height 8pt depth 0pt}}

\def\qed{\ifmmode\qquad\FullBox\else{\unskip\nobreak\hfil
\penalty50\hskip1em\null\nobreak\hfil\FullBox
\parfillskip=0pt\finalhyphendemerits=0\endgraf}\fi}

\def\qedsketch{\ifmmode\Box\else{\unskip\nobreak\hfil
\penalty50\hskip1em\null\nobreak\hfil$\Box$
\parfillskip=0pt\finalhyphendemerits=0\endgraf}\fi}

\newenvironment{proof}{\begin{trivlist} \item {\bf Proof:~~}}
  {\qed\end{trivlist}}
\newenvironment{proofsketch}{\begin{trivlist} \item {\bf Proof Sketch:~~}}
  {\qedsketch\end{trivlist}}

\newcommand{\eqdef}{\mathbin{\stackrel{\rm def}{=}}}

\newcommand{\zo}{\{0,1\}}
\newcommand{\zon}{\{0,1\}^n}
\newcommand{\zok}{\{0,1\}^k}

\newcommand{\pr}[1]{\Pr\left[#1\right]}

\newcommand{\E}{\mathop{\mathrm E}\displaylimits}

\newcommand{\eps}{\varepsilon}

\newcommand{\MyAtop}[2]{\genfrac{}{}{0pt}{}{#1}{#2}}

\newcommand{\be}{\begin{equation}}
\newcommand{\ee}{\end{equation}}
\newcommand{\h}[1]{\mathrm{height}(#1)}

\title{Sequential Rationality in Cryptographic Protocols}
\begin{document}

\ifnum\draft=2

\author{
\IEEEauthorblockN{Ronen Gradwohl}
\IEEEauthorblockA{Kellogg School of Management\\
Northwestern University\\
Evanston, IL, USA\\
Email: r-gradwohl@kellogg.northwestern.edu}
\and
\IEEEauthorblockN{Noam Livne}
\IEEEauthorblockA{Department of CS and Applied Mathematics\\
Weizmann Institute of Science\\
Rehovot, Israel\\
Email: noam.livne@weizmann.ac.il}
\and
\IEEEauthorblockN{Alon Rosen}
\IEEEauthorblockA{School of Computer Science\\
IDC Herzliya\\
Herzliya, Israel\\
Email: alon.rosen@idc.ac.il}
}

\else
\author{Ronen Gradwohl\thanks{Kellogg School of Management, Northwestern University, Evanston, IL 60208, USA. E-mail: \texttt{r-gradwohl@kellogg.northwestern.edu}. Work supported in part by ISF grant no. 334/08.}\and Noam Livne\thanks{Department of Computer Science and Applied Mathematics, Weizmann Institute of Science, Rehovot, 76100 Israel. E-mail: \texttt{noam.livne@weizmann.ac.il}. Work supported by ISF grant no. 334/08.} \and Alon Rosen\thanks{School of Computer Science, Herzliya IDC, 46150, Israel. E-mail: \texttt{alon.rosen@idc.ac.il}. Work supported by ISF grant no. 334/08.}}

\date{}

\begin{titlepage}
\fi

\maketitle

\ifnum\draft=2
\else
\thispagestyle{empty}
\fi

\begin{abstract}

Much of the literature on rational cryptography focuses on analyzing the strategic properties of cryptographic protocols. However, due to the presence of computationally-bounded players and the asymptotic nature of cryptographic security, a definition of sequential rationality for this setting has thus far eluded researchers.

We propose a new framework for overcoming these obstacles, and provide the first definitions of computational solution concepts that guarantee sequential rationality. We argue that natural computational variants of subgame perfection are too strong for cryptographic protocols. As an alternative, we introduce a weakening called threat-free Nash equilibrium that is more permissive but still eliminates the undesirable ``empty threats'' of non-sequential solution concepts.

To demonstrate the applicability of our framework, we revisit the problem of implementing a mediator for correlated equilibria (Dodis-Halevi-Rabin, Crypto'00), and propose a variant of their protocol that is sequentially rational for a non-trivial class of correlated equilibria. Our treatment provides a better understanding of the conditions under which mediators in a correlated equilibrium can be replaced by a stable protocol.

\end{abstract}

\ifnum\draft=2

\IEEEpeerreviewmaketitle
\else


\end{titlepage}
\fi

\tableofcontents

\section{Introduction}

A recent line of research has considered replacing the traditional cryptographic modeling of adversaries with a game-theoretic one. Rather than assuming arbitrary {\em malicious} behavior, participants are viewed as being self-interested, {\em rational} entities that wish to maximize their own profit, and that would deviate from a protocol's prescribed instructions if and only if it is in their best interest to do so.

Such game theoretic modeling is expected to facilitate the task of protocol design, since rational behavior may be easier to handle than malicious behavior. It also has the advantage of being more realistic in that it does not assume that some of the parties honestly follow the protocol's instructions, as is frequently done in cryptography.

The interplay between cryptography and game theory can also be beneficial to the latter. For instance, using tools from secure computation, it has been shown how to transform games in the mediated model into games in the unmediated model.

But regardless of whether one analyzes cryptographic protocols from a game theoretic perspective or whether one uses protocols to enhance game theory, it is clear that the results are meaningful only if one provides an adequate framework for such analyses.

\subsection{Computational Nash Equilibrium}
\label{CNE.sec}

Applying game-theoretic reasoning in a cryptographic context consists of modeling interaction as a {\em game}, and designing a protocol that is in {\em equilibrium}. The game specifies the model of interaction, as well as the utilities of the various players as a function of the game's outcome. The protocol lays out a specific plan of action for each player, with the goal of realizing some pre-specified task. Once a protocol has been shown to be in equilibrium, rational players are expected to follow it, thus reaching the desired outcome.

A key difficulty in applying game-theoretic reasoning to the analysis of cryptographic protocols stems from the latter's use of computational infeasibility. Whereas game theory places no bounds on the computational ability of players, in cryptography it is typically assumed that players are computationally bounded. Thus, in order to retain the meaningfulness of cryptographic protocols, it is imperative to restrict the set of strategies that are available to protocol participants. This gives rise to a natural analog of Nash equilibrium (NE), referred to as {\em computational Nash equilibrium} (CNE): any polynomial-time computable deviation of a player from the specified protocol can improve her utility by only a negligible amount (assuming other players stick to the prescribed strategy).

Consider, for example, the following (two-stage, zero-sum) game (related to a game studied by Ben-Sasson et al.~\cite{BKK07} and Fortnow
and Santhanam \cite{FS10}), which postulates the existence of a one-way permutation $f:\zon\mapsto\zon$.

\ifnum\draft=2
\begin{ex} (One-way permutation game):
\else
\begin{ex} {\bf (One-way permutation game):}
\fi
\label{ex:OWP-game}
\begin{enumerate}
\item $P_1$ chooses some $x\in\zon$, and sends $f(x)$.
\item $P_2$ sends a message $z\in\zon$.
\item $P_2$ wins (gets payoff 1) if $z=x$ (and gets -1 otherwise).
\end{enumerate}
\end{ex}
In classical game theory, in all NE of this game $P_2$ wins, since there always exists some $z$ such that $z=x$. However, in the computational setting, the following is a CNE: both players choose their messages uniformly at random (resulting in an expected loss for $P_2$). This is true because if $P_2$ chooses $z$ at random, then $P_1$ can never improve his payoff by not choosing at random. If $P_1$ chooses $x$ at random, then by the definition of a one-way permutation, any computationally-bounded strategy $\sigma_2$ of $P_2$ will be able to guess the value of $x$ with at most negligible (in $n$) probability. Thus, the expected utility of $P_2$ using $\sigma_2$ is
negligible, and so he loses at most that much by sticking to his CNE strategy (i.e.\ picking some $z$ at random).

\subsection{Computational Subgame Perfection}
\label{CSPE.sec}

The notion of CNE serves as a first stepping stone towards a game-theoretic treatment of cryptographic protocols. However, protocols are typically {\em interactive}, and CNE does not take their sequential nature into consideration. 

In traditional game theory interaction is modeled via
extensive games. The most basic equilibrium notion in this setting
is {\em subgame perfect equilibrium} (SPE),
which requires players' strategies to be in NE at any point of the interaction, regardless of the history of prior actions taken by other players. Basically, this ensures that players will not reconsider their actions as a result of reaching certain histories (a.k.a.~``empty threats").

As already noted in previous works (cf.\ \cite{K08,KN08b,OPRV09}), it is not at all clear how to adapt SPE to the computational setting. A natural approach would be to require the strategies to be CNE at every possible history. However, if we condition on the history, then this means that {\em different} machines can and will do much better than the prescribed equilibrium strategy. For example, in the one-way permutation game of Example~\ref{ex:OWP-game}, given any message history, a machine $M$ can simply have the correct inverse hardwired.

Although this requirement can be relaxed to ask that the prescribed strategy should be better than any other fixed machine on all inputs, this again may be too strong, since a fixed machine can always do better on some histories. Therefore, it seems that we must accept the following: for any machine $M$, with {\em high probability} over possible message histories, the prescribed strategy does at least as well as $M$. However, it turns out that this approach also fails to capture our intuitive understanding of a computational SPE (CSPE). Consider the following (two-stage) variant of the one-way permutation game from Example~\ref{ex:OWP-game}:

\ifnum\draft=2
\begin{ex} {(Modified one-way permutation game):}
\else
\begin{ex} {\bf (Modified one-way permutation game):}
\fi
\label{ex:modified-OWP-game}
\begin{enumerate}
\item $P_1$ chooses some $x\in\zon$, and sends $f(x)$.
\item $P_2$ sends a message $z\in\zon$.
\item If exactly one of $P_1$ and $P_2$ send message 0, both players get payoff $-2$. If both players
    send message 0, both players get payoff $+2$. Otherwise, $P_2$ wins (with payoff $+1$) if and
    only if $z=x$, and the non-winning player loses (with payoff $-1$).
\end{enumerate}
\end{ex}
Using a similar argument to the one applied in Section~\ref{CNE.sec}, it can be shown that the strategies in which both players choose a message uniformly at random from $\zon\setminus\{0\}$ satisfy the above ``probabilistic'' variant of CSPE. However, this equilibrium does not match our intuitive understanding of SPE: $P_1$ will prefer to send message 0 regardless of $P_2$'s strategy, knowing that $P_2$ will then respond with 0 as well. The threat of playing uniformly from all other messages is {empty}, and hence should not be admitted by the definition.\footnote{We note that a simple change to the payoffs yields a game whose empty threat is more ``typical'': For the case in which
both players send message 0, let $P_2$'s payoff be $-3/2$.}

The examples above are rather simple, so it is
reasonable to expect that issues arising in their analyses are inherent in many other cryptographic protocols.
This raises the question of whether a computational variant of
SPE is at all attainable in a cryptographic setting.

At the heart of this question is the fact that essentially any cryptographic protocol carries some small (but positive) probability of being broken. This means that, while there may be a polynomial-time TM that can ``perform well" on the {\em average} message history, there is no single TM that will do better than {\em all} other TMs on every history (as for any history there exists some TM that has the corresponding ``secret information" hardwired).

This state of affairs calls for an alternative approach. While such an approach should be meaningful enough to express strategic considerations in an interactive setting, it should also be sufficiently weak to be realizable. As demonstrated above, any approach for tackling this challenge should explicitly address the associated probability of error. It should also take asymptotics into consideration.

\section{Our Results}

We propose a new framework for guaranteeing sequential rationality in a computational setting. Our starting point is a weakening of subgame perfection, called {\em threat-free Nash equilibrium}, that is more permissive, but still eliminates the undesirable empty threats of non-sequential solution concepts.

To cast our new solution concept into the computational setting, we develop a methodology that enables us to ``translate" arguments that involve computational infeasibility into a purely game theoretic language. This translation enables us to argue about game theoretic concepts directly, abstracting away complications that are related to computation.

In order to demonstrate the applicability of our framework, we revisit the problem of implementing a mediator for correlated equilibria~\cite{DHR00}, and propose a protocol that is sequentially rational for a non-trivial class of correlated equilibria (see Section~\ref{sec:applications}
for details).
\ifnum\draft=2
\else
Our treatment provides a better understanding of the conditions under which mediators in a correlated equilibrium can be replaced by a stable protocol.
\fi

\ifnum\draft=2
We emphasize that this version of the paper is a summary of  \cite{GLR10a},
and we strongly recommend that the interested reader turn to the latter for proofs and more elaborate discussions.
\fi

\subsection{Threat-Free Nash Equilibria}

We introduce {\em threat-free Nash equilibria} (TFNE), a weakening of subgame perfection whose objective is to capture strategic considerations in an interactive setting. Loosely speaking, a pair of strategies in an extensive  game is a TFNE if it is a NE, and if in addition no player is facing an empty threat at any history.

The problem of empty threats is the following: in a NE of an extensive game, it is possible that a player plays sub-optimally at a history that is reached with probability 0. The other player may strategically choose to deviate from his prescribed strategy and arrive at that history, knowing that this will cause the first player to play an optimal response rather than the prescribed one. In an SPE this problem is eliminated by requiring that no player can play sub-optimally at any history, and so no other player will strategically deviate and take advantage of this.

The main observation leading to the definition of TFNE is that the above requirement may be too strong a condition to eliminate such instability: if an optimal response of a player {\em decreases} the utility of the other, then this other player would not want to strategically deviate. By explicitly ruling out this possibility, the instability caused by empty threats is eliminated, despite the equilibrium notion being more permissive than subgame perfection.

To make this precise, we give the first formal definition of an empty threat in extensive games. The definition is regressive: Roughly speaking, a player $i$ is facing a threat at a history if there is some deviation at that history, along with a threat-free continuation from that history onwards, so that $i$ increases his overall expected payoff when the players play this new deviation and continuation.

We note that the notion of TFNE is strong enough to eliminate the undesirable strategy of playing randomly in the modified OWP game from Example~\ref{ex:modified-OWP-game} --
\ifnum\draft=2
in the full version \cite{GLR10a} we show that in any computational TFNE of this game
the second player outputs 0 after history 0.
\else
Claim~\ref{claim:owp-tfne} shows that in any computational TFNE of this game
the second player outputs 0 after history 0.
\fi

\subsection{Strategy-Filters and Tractable Strategies}

To cast the definition of TFNE into a computational setting, we map the given protocol into a sequence of extensive games using {\em strategy-filters} that map computable strategies into their ``strategic representation" (the strategic representation corresponds to the strategy effectively played by a given interactive Turing machine). We can then apply pure game theoretic solution concepts, and in particular our newly introduced concept of TFNE, to understand the strategic behavior of players.

Similarly to the definition of CNE, the computational treatment departs from the traditional game theoretic treatment in two crucial ways. First of all, our definition is framed {\em asymptotically} (in order to capture computational infeasibility), whereas traditional game-theory is framed for finitely sized games. Second, it allows for a certain {\em error probability}. This is an artifact of the (typically negligible) probability with which the security of essentially any cryptographic scheme can be broken.

Given a cryptographic protocol, we consider a corresponding sequence of extensive games. The sequence is indexed by a security parameter $k$ and an error parameter $\eps$. For each game, we ``constrain'' the strategies available to players to be a subset of those that can be generated by PPT players in the protocol. Intuitively, the game indexed by $(k,\eps)$ contains those strategies that run in time polynomial in $k$ and ``break crypto" with probability at most $\eps$. We also require that strategy-filters be {\em PPT-covering}: that for any polynomially-small
$\eps$, every PPT is eventually a legal strategy, far enough into the sequence of extensive games.

Using this framework we formalize
\ifnum\draft=2
\else
the notion of a
\fi
computational threat-free Nash equilibrium (CTFNE).  To the best of our knowledge this is the first attempt at analyzing sequential strategic reasoning in the presence of
computational infeasibility.

\subsection{Applications}
\label{sec:applications}

Our treatment provides a powerful tool for arguing about the strategic behavior of players in a cryptographic protocol. It also enables us to isolate sequential strategic considerations that are suitable for use in cryptographic protocols (so that the solution concept is not too weak and not too strong).

\ifnum\draft=2
We
\else
As a warm up, we demonstrate the applicability of our framework and solution concept to the ``coin-flipping game'' that corresponds to Blum's coin-flipping protocol~\cite{Bl}. One may view this as playing the classic game of match pennies without simultaneity (but with cryptography).
We show that it is possible to exploit the specific structure of the game to implement a correlating device resulting in a CTFNE. This is in contrast to the general approach of \cite{DHR00} that only enables one to argue CNE. This result already demonstrates the added strength of our framework and definition.

We then
\fi revisit the general problem of implementing a mediator for correlated equilibria~\cite{DHR00}, and propose a protocol that is sequentially rational for a non-trivial class of correlated equilibria. In particular, our protocol is in a CTFNE for correlated equilibria that are convex combinations of Nash equilibria and that are ``undominated'': There does not exist any convex combination of Nash equilibria
for which both players get a strictly higher expected payoff.

Our treatment explores the conditions under which mediators in a correlated equilibrium can be replaced by a stable protocol, and sheds light on some structural properties of such equilibria.

Finally, we prove a general theorem that identifies sufficient conditions for a TFNE in extensive games. Namely, we show that if an
undominated NE has the additional property that no player can harm the other by a unilateral deviation, then that NE must also be threat-free.

\subsection{Related Work}
This paper contributes to the growing literature on rational cryptography. Many of the papers in this line of research,
such as
\cite{DHR00, HT04, IML05, ADGH06, GK06, LMS05, LT06, K08, KN08a, KN08b, KFN10, OPRV09, MS09, AL09, Gra10}, explore various solution concepts
for cryptographic protocols viewed as games (often in the context of rational secret-sharing).
Aside from the works of Lepinski et al.\ \cite{IML05,LMS05}, Ong et al.\ \cite{OPRV09}, and Gradwohl \cite{Gra10}, who work
in a different model\footnote{More specifically, \cite{IML05,LMS05} make strong physical assumptions,
 \cite{OPRV09} assume the existence of a fraction of honest (non-rational) players, and \cite{OPRV09, Gra10} work in an information-theoretic setting.},
all prior literature has considered solution concepts that are non-sequential. More specifically, they all use
variants of NE such as strict NE, NE with stability to trembles, and everlasting equilibrium.

An additional related work is that of Halpern and Pass \cite{HP10}, in which the authors present a general framework for game theory
in a setting with computational cost. While their approach to computational limitations is more general than ours, they
only address NE. Finally,  Fortnow and Santhanam \cite{FS10} study a different framework for games with computational limits,
but also only in the context of NE.

\subsection{Future Work}

One potential application of our new definition is an analysis of rational secret-sharing protocols.
\ifnum\draft=2
For some ideas about why known gradual release protocols satisfy a solution concept that is related to but slightly weaker than
CTFNE, see the full version of this paper \cite{GLR10a}.
\else
While the design of such a protocol that is in a CTFNE is not within the scope of the current paper, we do provide some intuition about why known gradual release protocols satisfy a slightly weaker solution concept. Consider the following simple setting: each of two players knows a bit, and the XOR of the two bits is the secret. Secret exchange protocols, for example~\cite{LMR83}, allow the players to exchange their respective bits and thus learn the secret in such a way that even if one of the players cheats, he can reconstruct the secret with probability at most $\eps$ more than the other player.
Then under the assumptions on players' utilities used by~\cite{KFN10}, any unilateral deviation from this protocol can get the deviating player an increase of only $O(\eps)$ in utility. However, since the other player can always correctly guess the secret with almost the same probability (up to the additive $\eps$), the potential benefit to a player of deviating, causing the other to deviate, and so on, is also at most $O(\eps)$. Thus, this protocol is in a computational variant of $\eps$-NE and is also $\eps$-threat-free.
The reason this is weaker than our current solution concept is that we require the benefit from a threat or a deviation to be negligible, whereas in~\cite{LMR83} the $\eps$ is polynomially-small (in the number of rounds of the protocol).
\fi

There are numerous other compelling problems left for future work. The first problem is to extend our definition to games with simultaneous moves. While we do offer a partial extension tailored to the problem of implementing a mediator, the problem of defining CTFNE for general games with simultaneous moves is open. Such a definition would be particularly useful for a sequential analysis of protocols with a simultaneous channel. Another natural extension of the definition is to multiple players, as opposed to 2. Such an extension comes with its own challenges, particularly with regard to the possibility of collusion. A third extension is to incorporate the threat-freeness property with stronger variants of NE, such as stability with respect to trembles, strict NE, or survival of iterated elimination of dominated strategies. Finally, we would like to find more applications for our definition. One particularly interesting problem is to extend our results on the implementation of mediators to a larger class of correlated equilibria.

\section{Game Theory Definitions}
\label{sec:defs-gt}

\subsection{Extensive Games}
Informally, a game in extensive form can be described as a game tree in which each node is owned
by some player and edges are labeled by legal actions. The game begins at the root,
and at each step follows the edge labeled by the action
chosen by the current node's owner. Utilities of players are given at the leaves of the tree. More formally, we have the following standard definition of extensive games (see, for example, Osborne and Rubinstein \cite{OR94}):

\begin{procdef}{Extensive game}
A 2-person {\em extensive game}  is a tuple $\Gamma=(H,P,A,u)$ where
\label{def:extensive}
\begin{itemize}
\item $H$ is a set of (finite) {\em history} sequences such that
the empty word $\epsilon\in H$. A history $h\in H$ is {\em terminal} if $\{a\,:\,(h,a)\in H\}=\emptyset$.
The set of terminal histories is denoted $Z$.
\item $P: (H\setminus Z)\rightarrow \{1,2\}$ is a function that assigns a ``next" player to every non-terminal history.
\item $A$ is a function that, for every non-terminal history $h\in H\setminus Z$, assigns
a finite set $A(h)= \{a\,:\,(h,a)\in H\}$ of available actions to player $P(h)$.
\item $u=(u_1,u_2)$ is a pair of payoff functions $u_i:Z\mapsto \mathbb{R}$.
\end{itemize}
\end{procdef}

We will denote the two players by $P_1$ and $P_2$ and by $P_i$ and $P_{-i}$, where $i\in\{1,2\}$ and
$-i$ is shorthand for $2-i$.

\begin{procdef}{Behavioral strategy}
{\em Behavioral strategies} of players in an extensive game are
collections $\sigma_i=\left(\sigma_i(h)\right)_{h:P(h)=i}$ of independent
probability measures,
where $\sigma_i(h)$
is a probability measure over $A(h)$.
\end{procdef}

For any extensive game $\Gamma=(H,P,A,u)$, any player $i$, and any history $h$
satisfying $P(h)=i$, we denote by $\Sigma_i(h)$ the set of all probability
measures over $A(h)$. We denote by $\Sigma_i$ the set of all strategies
$\sigma_i$ of player $i$ in $\Gamma$.
For each {\em profile} $\sigma=(\sigma_1,\sigma_2)$ of strategies, define the
{\em outcome} $O(\sigma)$ to be the probability distribution over terminal histories that results when
each player $i$ follows strategy $\sigma_i$.
\ifnum\draft=2
\else
Note that if both $\sigma_1$ and $\sigma_2$ are deterministic (i.e.\ deterministic on every history), then so is
the outcome $O(\sigma)$.
\fi

\subsection{Nash Equilibrium}

Each profile of strategies yields a distribution over outcomes, and we are interested in profiles
that guarantee the players some sort of optimal outcomes. There are many solution concepts that
capture various meanings of ``optimal,'' and one of the most basic is the Nash equilibrium (NE).

\begin{procdef}{Nash equilibrium (NE)}
An {\em $\eps$-Nash equilibrium} of an extensive game $\Gamma=(H,P,A,u)$ is a profile
$\sigma^*$ of strategies such that for each player $i$,
$$\E\left[u_i\left(O(\sigma^*)\right)\right] \geq \E\left[u_i\left(O(\sigma^*_{-i},\sigma_i)\right)\right]-\eps$$ for
every strategy $\sigma_i$ of player $i$. It is a {\em NE} if the above holds for $\eps\leq 0$
and a {\em strict NE} if it holds for some $\eps<0$.
\end{procdef}

One of the premises behind the stability of profiles that are in an $\eps$-NE is that players will not bother
to deviate for a mere gain of $\eps$. For applications in cryptography we will generally have $\eps$ be some negligible function,
and this corresponds to our understanding that we do not care about negligible gains.

\ifnum\draft=2
\else
\subsection{Subgame Perfect Equilibrium}

One of the problems with NE in extensive games is the presence of empty threats: a player's equilibrium strategy may specify a sub-optimal
strategy at a history that is reached with probability 0. The other player, knowing this, may strategically deviate to reach that history, predicting that the first
player will also deviate. For more details and explicit examples see any textbook on game theory, such as \cite{OR94}.

The most basic solution to the problem of empty threats is to refine the NE solution, and require a strategy profile to be in a NE at
every history in the game. This results in a profile that is in {\em subgame perfect equilibrium} (SPE).

\begin{definition}[Subgames of extensive game]
\label{def:subgame}
For any 2-person extensive game $\Gamma=(H,P,A,u)$  and any non-terminal history $h\in H$, the subgame
$\Gamma|_h$ is the 2-person extensive game $\Gamma|_h =(H|_h,P|_h,A|_h,u|_h)$, where
\begin{itemize}
\item $h'\in H|_h$ if and only if $h\circ h'\in H$,

\item $P|_h(h')=P(h\circ h')$,
\item $A|_h(h')=A(h\circ h')$, and
\item $u_i|_h(h')=u_i(h\circ h')$.
\end{itemize}
\end{definition}

For each profile $\sigma=(\sigma_1,\sigma_2)$ of strategies and history $h\in H$, define the
{\em conditional outcome} $O(\sigma)|_h$ to be the probability distribution over terminal histories that results when
the game starts at a history $h$, and from that point onwards each player $i$ follows strategy $\sigma_i$.

\begin{definition}[Subgame perfect equilibrium (SPE)]
\label{def:spe} An {\em $\eps$-subgame perfect equilibrium} of an
extensive game $\Gamma=(H,P,A,u)$ is a profile $\sigma^*$ of
strategies such that for each player $i$ and each non-terminal
history $h\in H$,
$$\E\left[u_i\left(O(\sigma^*)|_{h}\right)\right] \geq \E\left[u_i\left(O(\sigma^*_{-i},\sigma_i)|_{h}\right)\right]-\eps$$ for
every strategy $\sigma_i$ of player $i$. It is an {\em SPE} if the
above holds for $\eps=0$ and a {\em strict SPE} if it holds for some
$\eps<0$.
\end{definition}

\fi

\subsection{Constrained Games}
\label{sec:constrained-games}

In the standard game theory literature, where there are no
computational constraints on the players, the available strategies
$\sigma_i$ of player $i$ are all possible collections
$\left(\sigma_i(h)\right)_{h:P(h)=i}$, where $\sigma_i(h)$ is an
arbitrary distribution over $A(h)$.
In our setting, however, we will only consider strategies that can
be implemented by computationally bounded ITMs. This requires being
able to constrain players' strategies to a strict subset of the
possible strategies.
\ifnum\draft=2
\else
One natural way to restrict the strategies is to allow only a subset
of all distributions over $A(h)$ at each history $h$. However, this
does not enable us to capture more elaborate restrictions, and
specifically ones that might result from requiring strategies to be
implementable by polynomial time ITMs. (For example, a player
might have for every possible history a strategy that plays best
response on that history, but no strategy that plays best response
on {\em all} histories.) To capture these more elaborate
restrictions, we consider player $i$ strategies that are restricted
to an arbitrary subset $T_i$ of all possible (mixed) strategies.

\fi
Given a pair $T=(T_1,T_2)$ of such sets we can then define a constrained version of a game, in which only strategies that belong to these sets are considered.

\begin{procdef}{Constrained game}
\label{def:constrained-game} Let $\Gamma=(H,P,A,u)$ be an extensive game and let $T=(T_1,T_2)$, where $T_i\subseteq \bigotimes_{h:P(h)=i}\Sigma_i(h)$ for each $i\in \{1,2\}$. The {\em $T$-constrained version} of $\Gamma$ is the game in which the only allowed strategies for player $i$ belong to $T_i$.
\end{procdef}

\ifnum\draft=2
This definition enables us to capture restrictions that might result from requiring
strategies to be implementable by polynomial time ITMs.
\fi
NE of constrained games are defined similarly to regular NE, except that players' strategies
and deviations must be from the constraint sets.

\ifnum\draft=2
\else
\begin{definition}[NE in constrained games]
An {\em $\eps$-Nash equilibrium} of a $(T_1,T_2)$-constrained version of an extensive game $\Gamma=(H,P,A,u)$ is a profile
$\sigma^*\in (T_1,T_2)$ of strategies such that for each player $i$,
$$\E\left[u_i\left(O(\sigma^*)\right)\right] \geq \E\left[u_i\left(O(\sigma^*_{-i},\sigma_i)\right)\right]-\eps$$ for
every strategy $\sigma_i\in T_i$ of player $i$. It is a {\em NE} if the above holds for $\eps\leq 0$
and a {\em strict NE} if it holds for some $\eps<0$.
\end{definition}
\fi

\section{Threat-Free Nash Equilibrium}
\label{sec:tfne}

Our starting point is the inadequacy of subgame perfection in capturing sequential rationality in a computational context. As argued in Section~\ref{CSPE.sec}, it is unreasonable to require computationally-bounded players to play optimally at every node of a game. In particular, in cryptographic settings this requires breaking the security of the protocol, which is assumed impossible under the computational constraints.

A possible idea might be to require that players ``play optimally at
every node of the game, {under their computational
constraints}." However, this idea cannot be interpreted in a
sensible way. Computational constraints must be
defined ``globally," and thus the notion of
playing optimally under some computational constraint {on a
particular history} is senseless. In particular, for any
history of some cryptographic protocol, there is a small machine
that plays optimally on this specific history {\em unconditionally}
(and breaks ``cryptographic challenges" appearing in this
history, by having the solutions hardwired). This machine is
efficient, and so meets essentially any computational constraint.
So, while under computational constraints every machine fails on
cryptographic challenges in most histories, for every history there
is a machine that succeeds. We thus assume that a player chooses his machine before the game starts, and cannot change his machine~later.

\subsection{A New Solution Concept}

In light of the above discussion, it seems like the solution concept we are looking for has to reconcile the following seemingly conflicting properties:
\begin{enumerate}
\item It implies an optimal strategy for the players {\em under their computational constraints},
which implies {\em non-optimal} play on certain histories.
\item It does not allow empty threats, thus implying ``sequential rationality."
\end{enumerate}

The crucial observation behind our definition is that in order to rule out empty threats, one does not necessarily need to require
that players play optimally at {\em every} node, because not every
non-optimal play carries a threat to other players. In fact,
in a typical cryptographic protocol, the security of each player
is {\em building} on other players not playing
optimally (because playing optimally would mean breaking the
security of the protocol). Thus, a player's ``declaration" to play
non-optimally does not necessarily carry a threat: the other players
may even gain from it. More generally, even in non-cryptographic
protocols, at least in 2-player perfect information games, we can
use the following observation: in any
computational challenge, either a player gains from the other not
playing optimally, or, if he does not gain, he can avoid introducing
that computational challenge to the other player.\footnote{This is
indeed an informal statement. In fact, we should add the disclaimer
that computational hardness for one player does not necessarily have
to stem from the strategy of the other. For example, the utility
function may be computationally hard.}

Following the above observation, we introduce a new solution
concept for extensive games. The new solution concept requires that players be in NE, and moreover, that no player impose an empty threat on the other. At the same time, it does not require players to play optimally at every node. In other words, players may (declare to) play non-optimally on non-equilibrium support, yet this declaration of non-optimal play does not carry an empty threat. We call our new solution concept TFNE, for threat-free Nash equilibrium.

To make the above precise, we introduce a formal definition of an empty threat. An empty threat occurs when a player threatens to play ``non-rationally" on some history in order to coerce the other player to avoid this history. Crucially, empty threats are such that,
had the threatened \textit{not} believed the threat, had he deviated accordingly, and had the threatening player played ``rationally,'' the
threatened player would have benefitted. To rephrase our intuition: a player faces an empty threat with respect to some
strategy profile if by deviating from his prescribed strategy, and
having the other player react ``rationally," he improves his payoff (in comparison with sticking to the prescribed
strategy and having the other player react ``rationally'' from then on).

But what does it mean for the other player to react ``rationally"? The
other player may assume, recursively, that the first player will
play a best response, and will not carry out empty threats against him,
and so on, leading to a regressive definition.

\subsection{Vanilla Version}

Before giving the general definition of TFNE that we will use, we present a simpler version that has
no slackness parameter and that works for games without constrained strategies.

For a player $i$ and a history $h$, two strategies $\sigma_i$ and $\pi_i$ are {\em equivalent for player $i$ on $h$} if $P(h)=i$ and $\sigma_i(h)=\pi_i(h)$, or $P(h)\neq i$. Two strategies {\em differ
only on the subgame $h$} if they are equivalent on every non-terminal history that does not have $h$ as a prefix. Formally, they are equivalent on every
history in $H\setminus \{h'\in H:h'=h\circ h''\mbox{ for some
}h''\}$. For a history $h\in H$, a strategy $\sigma$, and a distribution $\tau=\tau(h)$ on $A(h)$,
\ifnum\draft=2
define the set $\mathrm{Cont}(h,\sigma,\tau)$ as
$$\Big\{\pi:
(\pi\; \mathrm{differs\; from}\; \sigma \;\mathrm{only \; on \;} h)\; \& \; (\pi(h)=\tau(h))\Big\}.$$
\else
let
\begin{align*}\mathrm{Cont}(h,\sigma,\tau)\eqdef\Big\{\pi:
(\pi\; \mathrm{differs\; from}\; \sigma \;\mathrm{only \; on \; the \; subgame}\; h)\; \& \; (\pi(h)=\tau(h))\Big\}.
\end{align*}
\fi
We now proceed to define a threat. For simplicity, we will do so for generic games, in which each player's possible payoffs are distinct. For such games, the set $\mathrm{Cont}(h,\sigma,\tau)$ always contains exactly one ``threat-free" element (defined below).

\begin{procdef}{Threat}
\label{def:vanilla-threat}
Let $\Gamma=(H,P,A,u)$ be an extensive game with distinct payoffs. Let $\sigma$ be a strategy profile, and let $h\in H$. Player $i=P(h)$ is facing a {\em threat} at history $h$ with respect to $\sigma$ if there exists a distribution $\tau=\tau(h)$ over $A(h)$ such that the unique $\pi\in \mathrm{Cont}(h,\sigma,\tau)$ and $\pi'\in \mathrm{Cont}(h,\sigma,\sigma)$ that are threat-free on $h$ satisfy \be\nonumber \E\left[u_{i}\left(O(\pi)\right)\right]
>
\E\left[u_{i}\left(O(\pi')\right)\right],\ee
where strategy $\pi$ is {\rm threat-free} on $h$ if for {\em all} $h'\not=\epsilon$ satisfying $h\circ h' \in H$ player $P(h\circ h')$ is not facing a threat at $h\circ h'$ with respect to $\pi$.
\end{procdef}
Note that if $h$ is such that for all $a\in A(h)$ it holds that $h\circ a \in Z$, then any profile $\pi$ is threat free on $h$.

\begin{procdef}{Threat-free Nash equilibrium}
\label{def:tfne}
Let $\Gamma=(H,P,A,u)$ be an extensive game. A strategy profile $\sigma^*$ is a {\em threat-free Nash equilibrium (TFNE)} if:
\begin{enumerate}
\item $\sigma^*$ is a $NE$ of $\Gamma$, and
\item for any $h\in H$, player $P(h)$ is not facing a threat at history $h$ with respect to $\sigma^*$.
\end{enumerate}
\end{procdef}

Note that in every profile that is in a TFNE, the effective play matches some SPE profile (more precisely, there is an SPE profile that yields exactly the same
distribution on outcomes). This and other properties of threats and TFNE are formalized in the companion paper to this work \cite{GLR10b}.

In the definition of a threat we used the fact that $\mathrm{Cont}(h,\sigma,\tau)$ and $\mathrm{Cont}(h,\sigma,\sigma)$
each contain exactly one profile that is threat-free on $h$. To
show that this must be the case, we have the following proposition\ifnum\draft=2
~(see full version \cite{GLR10a} for a proof)\fi, which is not unlike the fact that generic games have unique subgame perfect equilibria.

\begin{proposition}
\label{prop:well-defined-vanilla}{For any extensive game $\Gamma=(H,P,A,u)$, strategy profile
$\sigma$, player $i$, history $h\in H\setminus Z$ with
$P(h)=i$, and distribution $\tau$ over $A(h)$, the set $\mathrm{Cont}(h,\sigma,\tau)$ contains exactly one
profile that is threat-free on $h$.}
\end{proposition}

\ifnum\draft=2
\else
\begin{proof}
For any history $h\in H\setminus Z$, let $\h{h}$ be the maximal distance between $h$ and a descendant of $h$ (i.e.\ the leaf that is furthest away from $h$ but lies on the subtree rooted by $h$). The proof of the proposition is by induction on $\h{h}$.

For the base case $\h{h}=1$, note that there is exactly one element in
$\mathrm{Cont}(h,\sigma,\tau)$ and that this profile is threat-free on $h$ (since $h$ is a last move of the game).

Next, suppose the claim of the proposition holds for all histories $h$ with $\h{h}<k$. We will prove that it holds for histories $h$
with $\h{h}=k$. To this end, fix such a history $h^0$, and suppose the children of $h^0$ in the game tree are $h^1,\ldots,h^t$.
Suppose also that $P(h^0)=i$ and $P(h^1)=\ldots=P(h^t)=-i$, and note that this is without loss of generality.

Consider the profile $\pi^0$ that is identical to $\sigma$ except at history $h$, and fix $\pi^0(h)=\tau(h)$. We
now repeat the following process in succession for each $j\in\{1,\ldots,t\}$:
For any such $j$, let
$$\mathrm{TF}({h^j})\eqdef\left\{ \pi\in\bigcup_{\tau^j} \mathrm{Cont}(h^j,\pi^{j-1},\tau^j):
\pi\mbox{ is threat-free on }h^j \right\}.$$
We then choose a profile $\pi^j\in\mathrm{TF}({h^j})$ that satisfies
\begin{align*}
u_{-i}\left(\pi^j\right) \geq u_{-i}\left(\pi''\right)
\end{align*}
for all $\pi''\in\mathrm{TF}({h^j})$.
Because payoffs for player $-i$ are distinct, it must be the case that there exists a unique maximal $\pi^j$. That is, there
can be no $\pi''$ that is different from $\pi^j$ and has the same payoff for player $-i$.

After doing this for all $h^j\in\{h^1,\ldots,h^t\}$ we have a profile $\pi^t$ that we claim is threat-free on $h$.
To see this, observe that for all $j\in\{1,\ldots,t\}$, $\pi^j$ is threat-free on $h^j$ because we chose it to be a
threat-free profile from $\mathrm{Cont}(h^j,\pi^{j-1},\tau'')$. However, since for each $j$ we chose a \textit{maximal}
$\tau^j$, there are no threats at the histories $h^j$ either. Finally, uniqueness of $\pi^j$ is guaranteed by the fact that
for each $j$, our choice of a maximal $\tau^j$ was unique.
\end{proof}
\fi

\subsection{Round-Parameterized Version}
\label{sec:round-parametrized-version}

For games induced by cryptographic protocols we will need a more general definition of TFNE. We assume that in these games
players alternate moves, and thus there is a natural notion of the ``rounds'' in the game: Player $i$ makes a move in round 1, then player $-i$
makes a move in round 2, and so on until the end of the game.

For the general definition, we introduce a few modifications to the vanilla version:
\begin{itemize}
  \item  We add a slackness parameter $\eps$. This is necessary for our applications in order to handle the probability of error inherent in almost all cryptographic protocols.
  \item We allow players to be threatened at rounds, rather than just specific histories. This is needed because when we add the slackness parameter, a player might be threatened at a set of histories, where the weight of each individual threat does not exceed the slackness parameter, but the overall weight does.
  \item Finally, for a player to be threatened, we
require that he improve on {\em all} threat-free continuations
$\pi$. The reason we need this is that in the general case,
there may be more than one $\pi$ that is threat-free. If a
player deviates from his prescribed behavior, he cannot choose {\em which} (threat-free) continuation will be played.
\end{itemize}

The definitions below make use of the notion of a round $R$ strategy of player $i$: This is simply
a function mapping every history $h$ that reaches round $R$ to a distribution over $A(h)$. For a round
$R\in \mathbb{N}$ we let $\sigma_i(R)$ represent player $i$'s round $R$ strategy implied by $\sigma$.
Let $\sigma(R)=(\sigma_1(R),\sigma_2(R))$, and let
\begin{align*}\mathrm{Cont}(\sigma(1),\ldots,\sigma(R))\eqdef\Big\{\pi\in T: \pi(S) = \sigma(S) ~\forall S\leq R\Big\},
\end{align*}
where $T=(T_1,T_2)$ consists of constraints for players' strategies.

\begin{procdef}{$\eps$-threat}
\label{def:threat}
Let $\Gamma=(H,P,A,u)$ be an extensive game with constraints
$T=(T_1,T_2)$. Let $\eps\geq 0$, let $\sigma\in T$ be a strategy
profile, and let $R\in\mathbb{N}$ be a round of $\Gamma$. Player $i=P(R)$ is facing an {\em $\eps$-threat} at round $R$
with respect to $\sigma$ if there exists a round $R$ strategy
$\tau=\tau(R)$ for player $i$ such that
\begin{itemize}
\item[(i)] the set
$\mathrm{Cont}(\sigma(1),\ldots,\sigma(R\!-\!1),\tau(R))$ is nonempty,
and \item[(ii)] for all $\pi\in
\mathrm{Cont}(\sigma(1),\ldots,\sigma(R\!-\!1),\tau(R))$ and $\pi'\in
\mathrm{Cont}(\sigma(1),\ldots,\sigma(R))$ that are $\eps$-threat-free
on $R$
\begin{eqnarray*} \E\left[u_{i}\left(O(\pi)\right)\right] >
\E\left[u_{i}\left(O(\pi')\right)\right]+\eps,\end{eqnarray*}
\end{itemize}
where strategy $\pi$ is {\em $\eps$-threat-free on  $R$} if for
{\em all} rounds $S>R$ it holds that player $P(S)$ is not facing an
$\eps$-threat at round $S$ with respect to $\pi$.

\end{procdef}
Note that if $R$ is the last round of the game, then any profile $\pi\in
T$ is $\eps$-threat-free on $R$. Using Definition~\ref{def:threat}, we can now define an $\eps$-TFNE.

\begin{procdef}{$\eps$-threat-free Nash equilibrium}
\label{def:eps-tfne}
Let $\Gamma=(H,P,A,u)$ be an extensive game with constraints $T=(T_1,T_2)$.
A strategy profile $\sigma^*\in T$ is an {\em $\eps$-threat-free Nash equilibrium ($\eps$-TFNE)} if:
\begin{enumerate}
\item $\sigma^*$ is an $\eps$-NE of $\Gamma$, and
\item for any round $R$ of $\Gamma$, player $P(R)$ is not facing an $\eps$-threat at round $R$ with respect to $\sigma^*$.
\end{enumerate}
\end{procdef}

As is the case for Definition~\ref{def:vanilla-threat}, Definition~\ref{def:threat} (and hence Definition~\ref{def:eps-tfne}) would not be (semantically) well-defined if either one of the sets $\mathrm{Cont}(\sigma(1),\ldots,\sigma(R\!-\!1),\tau(R))$ or $\mathrm{Cont}(\sigma(1),\ldots,\sigma(R))$ would not contain at least one profile $\pi$ that is $\eps$-threat-free on $R$. The following proposition
shows that this can never be the case.

\begin{proposition}
\label{prop:well-defined}{Let $\Gamma=(H,P,A,u)$ be an extensive game with constraints
$T=(T_1,T_2)$. Let $\eps\geq 0$, let $\sigma\in T$ be a strategy
profile, and let $R$ be a round of $\Gamma$. For any round $R$ strategy $\tau=\tau(R)$ for player $i=P(R)$,
if the set $\mathrm{Cont}(\sigma(1),\ldots,\sigma(R\!-\!1),\tau(R))$ is nonempty then it contains at least one profile $\pi$ that is $\eps$-threat-free on $R$.}
\end{proposition}

\ifnum\draft=2
\else
\begin{proof}
For any round $R$ of $\Gamma$,  let $\h{R}$ be the distance between $h$ and the last round of $\Gamma$. The proof of the proposition is by induction on $\h{R}$.

For the base case $\h{R}=0$, note that, by the hypothesis of the proposition, the set $\mathrm{Cont}(\sigma(1),\ldots,\sigma(R\!-\!1),\tau(R))$ is nonempty. Since $R$ is the last round of the game, the set
contains exactly one profile, $(\sigma(1),\ldots,\sigma(R\!-\!1),\tau(R))$, and this profile is vacuously $\eps$-threat-free on $R$.

Next, suppose the claim of the proposition holds for all rounds $R$ with $\h{R}<k$. We will prove that it holds for round $R$
satisfying $\h{R}=k$. Let $i=P(R)$, and assume that there exists some $\pi'\in\mathrm{Cont}(\sigma(1),\ldots,\sigma(R\!-\!1),\tau(R))$.
We would like to show that $\mathrm{Cont}(\sigma(1),\ldots,\sigma(R\!-\!1),\tau(R))$ contains at least one profile $\pi$ that is $\eps$-threat-free on $R$.

By the inductive hypothesis we have that, for any round $R+1$ strategy $\tau'$ of player $-i$, if the set $\mathrm{Cont}(\sigma(1),\ldots,\sigma(R-1),\tau(R),\tau'(R+~1))$ is nonempty then it contains at least one profile that is $\eps$-threat-free on $R+1$  (since $\h{R+1}<k$). We will choose
a profile that has a \textit{maximal} $\tau'$ as follows.
Let
$$\mathrm{TF}({R+1})\eqdef\left\{ \pi\in\bigcup_{\tau'}\mathrm{Cont}(\sigma(1),\ldots,\sigma(R-1),\tau(R),\tau'(R+1)):
\pi\mbox{ is }\eps\mbox{-threat-free on }R+1 \right\},$$
and note that $\mathrm{TF}({R+1})$ must be nonempty. This is because there always exists at least one $\tau'$ for which $\mathrm{Cont}(\sigma(1),\ldots,\sigma(R-1),\tau(R),\tau'(R+1))$ is nonempty: namely,
we could have $\tau'(R+1)=\pi'(R+1)$. Since $\mathrm{Cont}(\sigma(1),\ldots,\sigma(R-1),\tau(R),\pi'(R+1))$ is nonempty by assumption, it must contain a profile that is $\eps$-threat-free on $R+1$
(by the inductive hypothesis).

We now choose a profile $\pi\in\mathrm{TF}({R+1})$ that satisfies
\begin{align*}
u_{-i}\left(\pi\right) \geq u_{-i}\left(\pi''\right)-\eps
\end{align*}
for all $\pi''\in\mathrm{TF}({R+1})$.
So now we have a profile $\pi \in \mathrm{Cont}(\sigma(1),\ldots,\sigma(R\!-\!1),\tau(R))$, which we claim is
$\eps$-threat-free on round $R$. To see this, note that $\pi$ is $\eps$-threat-free on $R+1$ by the way we chose it (i.e.\ a profile from $\mathrm{Cont}(\sigma(1),\ldots,\sigma(R-1),\tau(R),\tau'(R+1))$ that
is $\eps$-threat-free on $R+1$). However, since we chose a \textit{maximal} $\tau'$ (up to $\eps$), there is no $\eps$-threat at round  $R+1$ either. Thus $\pi$ is $\eps$-threat-free on $R$.
\end{proof}
\fi

\section{The Computational Setting}

In the following we explain how to use the notion of TFNE for cryptographic protocols. In Section \ref{sec:defs-crypto} we describe how to view a cryptographic protocol as a sequence of extensive games. In Section \ref{ITM} we show how to translate the behavior of an interactive TM to a sequence of strategies. In Section \ref{express} we show how to express computational hardness in a game-theoretic setting. Finally, in Section \ref{sec:comp-tfne} we give our definition of computational TFNE.

\subsection{Protocols as Sequences of Games}
\label{sec:defs-crypto}

When placing cryptographic protocols in the framework of
extensive games, the possible messages of players in a protocol
correspond to the available actions in the game tree, and the
prescribed instructions correspond to a strategy in the game.

The protocol is parameterized by a security parameter $k\in
\mathbb{N}$. The set of possible messages in the protocol, as well
as its prescribed instructions, typically depend on this $k$.
Assigning for each $k$ and each party a payoff for every outcome, a
protocol naturally induces a sequence
$\Gamma^{(k)}=(H^{(k)},P^{(k)},A^{(k)},u^{(k)})$ of extensive
games, where:
\begin{itemize}
\item $H^{(k)}$ is the set of possible {\em transcripts} of the protocol (sequences of messages exchanged between the parties). A history $h\in H^{(k)}$ is {\em terminal} if the prescribed instructions of the protocol instruct the player whose turn it is to play next to halt on input $h$. 
\item $P^{(k)}: (H^{(k)}\setminus Z^{(k)})\rightarrow \{1,2\}$ is a function that assigns a ``next" player to every non-terminal history.
\item $A^{(k)}$ is a function that assigns to every non-terminal history $h\in H^{(k)}\setminus Z^{(k)}$ a set $A^{(k)}(h)= \{m\,:\,(h \circ m)\in H^{(k)}\}$ of possible protocol messages to player $P^{(k)}(h)$.\footnote{We can interpret ``disallowed" messages in the protocol as abort, and define ``abort" as a possible protocol message. This will imply that every execution of the protocol corresponds to some history in the game.}

\item $u^{(k)}=(u_1^{(k)},u^{(k)}_2)$ is a vector of payoff functions
$u^{(k)}_i:Z^{(k)}\rightarrow \mathbb{R}$.
\end{itemize}
A sequence $\Gamma=\{\Gamma^{(k)}\}_{k \in \mathbb{N}}$ of games defined as above is referred to as a {\em computational game}.

\begin{remark}
In the following we will consider games played by Turing machines. Thus, actions will be represented by strings. As opposed to traditional game theory, where players are computationally unbounded, in our case the names of the actions will be significant. For example, in the One-way Permutation Game, if we encode player 1's action $f(x)$ by the string $x$ for every $x \in \zok$, then inverting the one-way permutation becomes easy for player 2. However, to avoid too much notation, we will identify actions with their string representation. The reader should keep in mind, however, that actions are always strings, and that changing the string representation of actions might be {\em with} loss of generality.
\end{remark}

\subsection{Strategic Representation of Interactive Machines} \label{ITM}

Protocols are defined in terms of {\em interactive Turing machines} (ITMs) -- see \cite{GFoC} for a formal definition. More specifically, the prescribed behavior for each player is defined via an ITM, and any possible deviation of this player corresponds to choosing a different ITM. In order to argue about the protocol in a game-theoretic manner we formalize, using game-theoretic notions, the strategic behavior implied by ITMs. We believe this formalization is necessary for our treatment or any game-theoretic analysis of ITMs, in particular because, to the best of our knowledge, it has never been done before.
\ifnum\draft=2
The full formalization is deferred to \cite{GLR10a}, and has the following (informally stated) conclusion:
\else
However, because this section somewhat departs from the main thrust of the paper, the reader may skip to Section~\ref{express},
keeping the following (informally stated) conclusion in mind:
\fi
 The strategic behavior of an ITM for player $i$ in a protocol may be seen as a collection of independent distributions on actions, one for each of player $i$'s histories that are reached with positive probability given the ITM of player $i$ and some strategy profile of the other players. We refer to this collection as the behavioral reduced strategy induced by the ITM.

\ifnum\draft=2
\else
When considering some computational game $\Gamma^{(k)}$ in a sequence $\Gamma=\{\Gamma^{(k)}\}_{k \in \mathbb{N}}$ and an ITM ``playing" this game (with input $1^k$), the machine does not, strictly speaking, define a strategy. Informally, the machine specifies how to play {\em only on histories that are not inconsistent with the specification on earlier histories in the game}. That is, an ITM for player $i$ specifies distributions on actions for all histories on which it is player $i$'s turn, except those it cannot reach based on its own specification on earlier histories. This is the case, because when fixing the other player's moves, the distribution on actions the machine plays on a history that cannot be reached is simply undefined, as we are conditioning on an event with probability $0$. In the following, we show that the prescribed behavior of an ITM can be seen as a convex combination of {\em reduced strategies} (which we call {\em mixed reduced strategy}), to be defined next. We then define the natural analogue of {\em behavioral reduced strategy}, and argue that for every mixed reduced strategy there exists a behavioral reduced strategy that is outcome-equivalent. We will eventually use  behavioral reduced strategies to describe the behavior induced by ITMs.

\begin{procdef}{Reduced strategy (adapted from \cite{OR94})}
Given a game $\Gamma=(H,P,A,u)$, a (pure) reduced strategy for player $i$ is a function $\sigma_i$ whose domain is a subset of $\{h\in H | P(h)=i\}$ with the following properties:
\begin{itemize}
  \item For every $h$ in the domain of $\sigma_i$ it holds that $\sigma_i(h) \in A(h)$.
  \item $h=(a_1,\dots,a_m)$ is in the domain of $\sigma_i$ if and only if for any $1 \leq \ell \leq m-1$ such that $P(a_1,\dots,a_\ell)=i$ it holds that $(a_1,\dots,a_\ell)$ is in the domain of $\sigma_i$ and $\sigma_i(a_1,\dots,a_\ell)=a_{\ell+1}$.
\end{itemize}
\end{procdef}

\begin{definition}[Mixed reduced strategy]
A {\em mixed reduced strategy} for player $i$ is a distribution over reduced strategies for player $i$.
\end{definition}

Given an ITM for $\Gamma^{(k)}$, for every instance of internal randomness for that machine (i.e., a vector of coins), the induced behavior of that ITM is exactly a reduced strategy. This is the case because for every profile of pure strategies (or reduced pure strategies) of the other players, the randomness naturally defines an action for every history that is consistent with its previous actions (the sequence of these actions, together with the profile, defines the outcome of the game), and on the other hand, naturally the randomness does not define an action for histories that are not consistent with that randomness (as with that randomness the machine will never reach these histories). It follows that an ITM defines a distribution over reduced (pure) strategies, i.e., a mixed reduced strategy. We now formalize this claim.

\begin{definition}[Induced mixed reduced strategy of an ITM]
Let $M$ be a probabilistic ITM for player $i$ in the extensive game $\Gamma$. Assume that $M$ halts for any infinite vector of coins and any sequence of messages sent by the other players, and let $t$ be a bound on the number of coins it reads. Let $r$ be a (sufficiently long) coin vector for $M$. Then the induced pure reduced strategy $\sigma^{(r)}_i$ of $M$ with randomness $r$ is defined as follows:
\begin{itemize}
\item $h=(a_1,\dots,a_m)$ is in the domain of $\sigma^{(r)}_i$ if and only if:
\begin{itemize}
  \item $P(a_1,\dots,a_m)=i$;
  \item For any $1 \leq \ell \leq m-1$ such that $P(a_1,\dots,a_\ell)=i$ it holds that $(a_1,\dots,a_\ell)$ is in the domain of $\sigma^{(r)}_i$ and when $M$ with randomness $r$ participates in an interaction, conditioned on the sequence of sent messages being $(a_1,\dots,a_\ell)$ (where $a_{\ell+1}$ is a message sent by the ITM representing player $P(a_1,\dots,a_\ell)$ for any $1 \leq \ell \leq m-1$), the message sent by $M$ is $a_{\ell+1}$.\footnote{For completeness, we may assume that whenever $M$ outputs on history $h$ an action that is not in $A(h)$, we interpret it as abort, which is denoted in the induced game by $\bot$ and is always a legal action.}
\end{itemize}
\item For any $h=(a_1,\dots,a_m)$ in the domain of $\sigma^{(r)}_i$, the action $\sigma^{(r)}_i(a_1,\cdots,a_m)$ is the message sent by $M$ with randomness $r$ conditioned on the sequence of sent messages being $(a_1,\dots,a_m)$.
\end{itemize}

The {\em mixed reduced strategy induced by $M$} is now defined as follows: the probability assigned to any pure reduced strategy $\sigma$ is the probability that the induced reduced strategy of $M$ with randomness $r$ is $\sigma$, where $r$ is uniformly chosen from $U_t$.
\end{definition}

In \cite{OR94} it is shown that for perfect-recall extensive games (which are the only games we will consider here), every mixed strategy has a behavioral strategy that is outcome equivalent. (Two strategies are outcome-equivalent if for every profile of pure strategies of the other players the two strategies induce the same distribution on outcomes; A mixed strategy is a distribution on pure strategies). Next, we define the behavioral analogue of a mixed reduced strategy, and argue that the same holds for mixed and behavioral {\em reduced} strategies: For perfect-recall extensive games, every mixed reduced strategy has a behavioral reduced strategy that is outcome equivalent.

\begin{definition}[Behavioral reduced strategy]
Given a game $\Gamma=(H,P,A,u)$, a behavioral reduced strategy for player $i$ is a collection $\sigma_i=\left(\sigma_i(h)\right)_{h \in {\cal H}}$ of independent
probability measures, where ${\cal H}$ is a subset of $\{h\in H | P(h)=i\}$, with the following properties:
\begin{itemize}
  \item $\sigma_i(h)$ is a probability measure over $A(h)$ for every $h$ in ${\cal H}$.
  \item $h=(a_1,\dots,a_m)$ is in ${\cal H}$ if and only if for any $1 \leq \ell \leq m-1$ such that $P(a_1,\dots,a_\ell)=i$ it holds that $(a_1,\dots,a_\ell) \in {\cal H}$ and $\sigma_i(a_1,\dots,a_\ell)(a_{\ell+1}) > 0$.
\end{itemize}
\end{definition}

\begin{claim} \label{mix-beh}
Every mixed reduced strategy has a behavioral reduced strategy that is outcome equivalent.
\end{claim}

\begin{proofsketch}
Every pure reduced strategy $\sigma_i$ for player $i$ can be extended to a (full) pure strategy by assigning arbitrary values to all histories in $\{h: P(h)=i\}$ for which $\sigma_i$ is undefined. The two strategies will be outcome-equivalent, as the outcome is only affected by the consistent histories of $\sigma_i$. It follows that every mixed reduced strategy can be extended to a mixed (full) strategy that is outcome-equivalent.

On the other hand, every behavioral strategy $\sigma_i=\left(\sigma_i(h)\right)_{h:P(h)=i}$ can be restricted to a behavioral reduced strategy by restricting the collection of probability measures accordingly. Again, the two strategies will be outcome-equivalent, as the distribution on outcomes is only affected by the consistent histories of $\sigma_i$.

Finally, as mentioned above, in \cite{OR94} it is shown that for perfect-recall extensive games, every mixed strategy has a behavioral strategy that is outcome equivalent.

Thus, given some mixed reduced strategy we extended it to a mixed strategy that is outcome-equivalent, then transform it to a behavioral strategy that is outcome-equivalent, and finally we restrict the resulting behavioral strategy to an outcome-equivalent behavioral reduced strategy.
\end{proofsketch}

As argued above, ITMs induce mixed reduced strategies, and by Claim \ref{mix-beh}, these induce behavioral reduced strategies. Thus, in the following we will model ITMs by behavioral reduced strategies. This is captured by the notion of {\em strategic representation}.

\begin{definition}[Strategic representation of an ITM]
Let $\Gamma$ be a game and let
$i\in\{1,2\}$. Let $M$ be an ITM for player $i$. Assume that $M$ halts for any infinite vector of coins and any sequence of messages sent by the other players. Let $\sigma$ be the mixed reduced strategy induced by $M$. Then the {\rm strategic representation} of $M$ is
the behavioral reduced strategy that is outcome-equivalent to $\sigma$.\footnote{In certain games there may be more than one behavioral reduced strategy that is outcome-equivalent to $\sigma$. However, our treatment will always be indifferent to the actual choice.}

Similarly, for a sequence of games $\{\Gamma^{(k)}\}_{k\in\mathbb{N}}$ and an ITM $M$ that takes a security parameter $1^k$, the strategic representation of $M$ is the sequence of strategic representations of $M(1),M(1^2),M(1^3),\dots$.
\end{definition}

\subsubsection{$\eps$-TFNE for Reduced Strategies}
\label{sec:tfne-reduced}
In Section~\ref{sec:round-parametrized-version} we presented our general definition of TFNE. However, that definition
was framed for strategies and, following the conclusion of the previous section, we actually care about reduced strategies.
To make Definition~\ref{def:eps-tfne} work for reduced strategies we notice that only two small changes need to be made:
We need to define the notion of a round $R$ reduced strategy, and we need to allow the constraint sets $T_1$ and $T_2$ to include
behavioral reduced strategies.

\begin{definition}[Round $R$ reduced strategy]
\label{def:round-R-strategy}
Let $\Gamma=(H,P,A,u)$ be an extensive game, let $R$ be a round of $\Gamma$,
and let $\sigma_i$ be a behavioral reduced strategy of player $i=P(R)$. Then $\tau=\tau(R)$
is a {\em round $R$ reduced strategy} of player $i$ consistent with $\sigma_i$ if the following hold:
\begin{itemize}
\item When $R=1$, $\tau(1)$ is a distribution over $A(\epsilon)$.
\item Otherwise, there exists some
behavioral reduced strategy $\pi_i$ of player $i$ for which $\pi_i(j)=\sigma_i(j)$ for all $j\in\{1,\ldots,R\!-\!1\}$, and such that
$\pi_i(R)=\tau_i(R)$.
\end{itemize}
\end{definition}
Throughout the paper, the behavioral reduced strategy $\sigma_i$ with which $\tau(R)$ is consistent will be evident from
the context, and so we omit reference to this consistency requirement.

Next, we modify the definition of constraints (Definition~\ref{def:constrained-game}) by allowing each constraint set $T_i$
to be a subset of $\bigotimes_{h:P(h)=i}(\Sigma_i(h)\cup\perp)$, where $\sigma_i(h)=\perp$ if the history $h$ is not in
the domain of the reduced strategy $\sigma_i$.

Finally, we observe that, following the two modifications above, Definitions~\ref{def:threat} and~\ref{def:eps-tfne}
work for behavioral reduced strategies as well (replacing ``strategy'' by ``behavioral reduced strategy'' and ``round $R$ strategy'' by
``round $R$ reduced strategy'').
\fi

\subsection{Computational Hardness in the Game-Theoretic Setting} \label{express}

The security of cryptographic protocols stems from the assumption on the limitation of the computational power of the players. In our strategic analysis of games, we also expect to deduce the (sequential) equilibrium from this limitation. However, because protocols are parameterized by a security parameter, a strategic analysis of protocols requires dealing with a {\em
sequence} of games rather than a single game. While
relating to the sequence of games is crucial in order to express
computational hardness (as this hardness is defined in an asymptotic
manner), this raises a new difficulty: How do we
extend the definition of TFNE to sequences of games?

\ifnum\draft=2
Our
\else
An appealing approach might be to try to define empty threats for
sequences of games. That is, one might consider the effect of deviations on
the expected payoff as $k$ goes to infinity (much like the
derivation of CNE from NE). However, to the best of our
understanding this approach cannot work. Loosely speaking, this is
because in order to relate to empty threats one has to consider
deviations in internal nodes of the game tree, and it is not clear
how to define such deviations for sequences of games. Typically, the  structure of the game tree changes with $k$, so it is
not clear even how to define an ``internal node'' in a \textit{sequence} of
games.

Instead, our
\fi
approach insists on analyzing empty threats for {\em
individual} games. Thus, our solution concept reflects a hybrid approach that relates to a protocol both as a family of {\em individual, extensive games} and as a {\em sequence} of {\em normal-form games}. To eliminate empty threats one
must relate to the {\em interactive} aspect of each {\em individual}
game (as this is the setting where threats are defined). In order to
claim players are playing optimally under their computational
constraints, one must think of the protocol as a {\em sequence} of
{\em one-shot} games (because computational hardness is meaningful only when players are required to choose their machines in advance, and as the traditional notion of hardness is stated asymptotically).

\subsubsection{Strategy-filters}

When considering computational games $\Gamma=\{\Gamma^{(k)}\}_{k \in \mathbb{N}}$, the computational bounds on the players will be expressed by restricting the space of available strategies for the players. The available sequences of reduced strategies for the players will be exactly those that can be played by the ITMs that meet the computational bound on the players. In our case we will consider PPT ITMs.

While on the one hand every PPT ITM fails on cryptographic
challenges for large enough values of the security parameter $k$
(under appropriate assumptions), on the other hand, PPT ITMs
can have arbitrarily large size and thus
arbitrarily much information hardwired, and so for every $k$ there is a PPT ITM that breaks the cryptographic challenges with security parameter
$k$. In our analysis, we would like to ``filter" machines according
to their ability to break cryptographic challenges for specific
$k$'s, and allow using them only in games that correspond to large
enough $k$'s, where these machines fail (and in particular, cannot
use hard-wiring to solve the cryptographic challenges).

To this end, we define the notion of a {\em strategy-filter}. For each value $k$ of the security parameter and value $\eps$, a strategy-filter maps the ITM $M$ to either $\bot$ or to its strategic representation, according to whether $M(1^k)$ violates level of security $\eps$ or does not (respectively).

\begin{procdef}{Strategy-filter}
Let $\Gamma=\{\Gamma^{(k)}\}_{k \in \mathbb{N}}$ be a computational
game and let $i$ be a player. A {\em strategy-filter} is a sequence
$F_i=\{F_i^{(k)}: {\cal M}\times [0,1] \rightarrow
\Sigma^{(k)}_i\cup \{ \bot \}\}_{k\in \mathbb{N}}$ such that for
every ITM $M$, every $k \in \mathbb{N}$ and every $\eps \in [0,1]$,
it holds that either $F_i^{(k)}(M,\eps)=\bot$, or
$F_i^{(k)}(M,\eps)=\sigma_i^{(k)}$, where $\sigma_i^{(k)}$ is the
strategic representation of the machine $M(1^k,\cdot)$.
\end{procdef}

A strategy-filter is meaningful if it allows us to reason about all reduced strategies that are considered to be feasible, in our case PPT implementable reduced strategies, and in particular does not filter them out.
\ifnum\draft=2
\else
This is captured in the following definition.
\fi

\begin{procdef}{PPT-covering filter}
A strategy-filter $F_i$ is said to be {\em PPT-covering} if for
every PPT ITM $M$ and any positive polynomial $p(\cdot)$ there exists $k_0$ such that for all $k \geq k_0$, it holds that $F_i^{(k)}(M,1/p(k))\neq \bot$.
\end{procdef}

Typically, protocols have the following security guarantee (under computational assumptions): for every $i$, every PPT ITM $M$ of $P_i$ and every polynomial $p(\cdot)$, there exists $k_0$ such that for any $k \geq k_0$, the ITM $M$ does not break level of security $1/p(k)$ in the protocol with security parameter $k$. Such a protocol will naturally have a PPT-covering filter, where if $F_i^{(k)}(M,\eps)\neq \bot$ then the reduced strategy $F_i^{(k)}(M,\eps)$ ``does not break level of security $\eps$ in the game $\Gamma^{(k)}$."

\subsubsection{Tractable Reduced Strategies}

As reflected above, the asymptotic nature of defining security does not determine any level of security for any $k$. Rather, it dictates that any PPT ITM ``eventually fails in violating $1/p(k)$ security" for any $p(\cdot)$ (where ``eventually" means for large enough $k$). Thus, we follow the same approach in our game theoretic analysis: roughly speaking, our solution concept requires that $\eps$-security will imply $\eps$-stability for any $k$ (rather than requiring a particular level of stability for each $k$). More formally, we require that for any $k$ and any $\eps$, the game induced by the protocol with security parameter $k$ be in $\eps$-TFNE, given that the available strategies for the players are those that do not break level of security $\eps$. Thus, for any pair $(k,\eps)$ we will consider the game $\Gamma^{(k)}$ with available reduced strategies restricted to those that guarantee $\eps$-security. The following definition derives from a PPT-covering filter, for each such game, the set of available reduced strategies for each player.

\begin{procdef}{Tractable reduced strategies} Let $F_i$ be a PPT-covering filter. For every $k\in \mathbb{N}$ and $\eps \in [0,1]$ we define the set $T^{(k)}_{i,\eps}(F_i)$ of $(k,\eps)$-tractable reduced strategies for player $i\in\{1,2\}$ as
$$\{F_i^{(k)}(M,\eps)|M ~\mbox{is a PPT ITM
and}~F_i^{(k)}(M,\eps)\neq\bot\}.$$
\end{procdef}
Whenever $F_i$ will be understood from the context, we will write $T^{(k)}_{i,\eps}$ to mean $T^{(k)}_{i,\eps}(F_i)$.

\subsection{Computational TFNE}
\label{sec:comp-tfne}

We can now define our computational variant of TFNE. Roughly, the definition requires that
 there exist a family
of PPT compatible constraints such that for any $k$ and any
$\eps$, the strategies played by the machines on input security
parameter $k$ are in $\eps$-TFNE in the game indexed by $(k,\eps)$.

\begin{procdef}{Computational TFNE}\label{def:comp-tfne}
Let $\Gamma$
be a computational game. A pair of PPT machines $(M_1,M_2)$ is said to be in a {\em
computational threat-free Nash equilibrium (CTFNE)} of $\Gamma$ if
there exists a pair of PPT-covering filters $(F_1,F_2)$ such that
for every $k,\eps$
 for which 
$F_1^{(k)}(M_1,\eps)$ and $F_2^{(k)}(M_2,\eps)$ are tractable the
profile
 $(F_1^{(k)}(M_1,\eps),F_2^{(k)}(M_2,\eps))$ constitutes an $\eps$-TFNE in the
$(T^{(k)}_{1,\eps}, T^{(k)}_{2,\eps})$-constrained version of
$\Gamma^{(k)}$.
\end{procdef}

\ifnum\draft=2
\else
The expressive power of Definition~\ref{def:comp-tfne} is illustrated through the following claim, which refers to Example~\ref{ex:modified-OWP-game}. We omit the proof, and proceed to more interesting applications in sections \ref{app:proof-chicken} and \ref{sec-DHR}.
\begin{claim}
\label{claim:owp-tfne}
In the modified one-way permutation game, \begin{itemize} \item[(i)] the strategy profile in which $P_1$ plays 0 and $P_2$ plays 0 after a history of 0 and randomly otherwise is a CTFNE, and \item[(ii)] any profile in which $P_2$ plays randomly after history 0 is not a CTFNE.
\end{itemize}
\end{claim}

We note that part (ii) of the claim can easily be extended to profiles in which, after history 0, $P_2$ plays 0 with probability at most $1-p(k)$ for any polynomial $p$.

\fi

\ifnum\draft=2
\else

\section{The Coin-Flipping Game}
\label{app:proof-chicken}

In the following we describe a classic protocol for coin-flipping,
formulated as a sequence of games (parameterized by a security
parameter $k$). We then show that the prescribed behavior according
to that protocol constitutes a CTFNE in the sequence of games.

Following is an informal description of the sequence of games. We
assume some perfectly binding commitment scheme with the following
properties (see Appendix~\ref{apx:defs} for a formal definition):
\begin{itemize}
\item For any security parameter $k$ (which is a common input to the sender and receiver), the ``commit" phase consists of one message from the sender to the receiver, denoted ${\sf
  com}^{(k)}$, which is of length bounded by $p(k)$ for some polynomial $p$.
\item For any PPT ITM, the advantage in guessing the committed value given the aforementioned message is
  negligible in $k$.
\end{itemize}

\noindent The description defines the legal messages in each game.
Recall that at any phase where a player is supposed to send a
message, the move ``abort" is legal (and well-defined). Note also,
that any illegal message is interpreted as abort by the other
player. The game $\Gamma^{(k)}$ is defined as follows:

\begin{enumerate}
\item Player 1 chooses a string $c$ of length at most $p(k)$ and sends it to
player 2.
\item Player 2 chooses a bit $r_2$, and sends $r_2$ to player
1.
\item Player 1 does one of the following: (1) sends to player 2 ${\sf decom}$,
where ${\sf decom}$ is a legal decommitment to $c$ revealing that
the committed value was $1-r_2$ (in that case the payoffs are
(1,0)); or (2) aborts (in that case
 the payoffs are (0,1)).
\end{enumerate}
Any other abort results in the aborting player receiving payoff
0, and the other player receiving 1.

We now describe a pair of interactive ITMs for the game
$\Gamma^{(k)}$ that form a CTFNE. We describe them interleaved, in
the form of a protocol. We denote the ITMs playing the strategies of
$P_1,P_2$ by $M_1,M_2$, respectively.
\begin{enumerate}
\item Player 1 chooses a random bit $r_1$, and sends $c={\sf com}^{(k)}(r_1)$ to
player 2 (player 1 also obtains ${\sf decom}$, which is a legal
decommitment to $c$).
\item Player 2 chooses a random bit $r_2$, and sends $r_2$ to player
1.
\item If $r_1 \neq r_2$, player 1 sends ${\sf decom}$ to
player 2. Else, player 1 aborts.
\end{enumerate}
\begin{theorem}
\label{thm:coin-flipping-tfne}
The pair $(M_1,M_2)$ forms a CTFNE for the protocol above.
\end{theorem}
\begin{proof}
First we define the functions $F_1^{(k)}$ and $F_2^{(k)}$. For any $k$,
the function $F_1^{(k)}$ never maps to $\bot$ (this, roughly
speaking, reflects the fact that the protocol is secure against an
all-powerful player 1). For $F_2$ we use the following rule:
$F_2^{(k)}(M,\eps)=\bot$ if and only if ``for security parameter $k$,
the PPT ITM $M$ guesses the committed value with advantage
greater than $\eps$." More formally, $F_2^{(k)}(M,\eps)=\bot$ if and only if  when player 1 sends as the
first message a random commitment of a random bit (i.e., chooses a
random bit and then uses the aforementioned commitment scheme using
uniformly random coins), then the message with which $M$ reacts is the
committed value of player 1 with probability greater than
$1/2+\eps$.

The fact that $F_1$ is PPT-covering is straightforward. The fact
that $F_2$ is PPT covering follows directly from the security of the
commitment scheme: For any positive polynomial
$p$, every PPT ITM has advantage smaller than $1/p(k)$ in
guessing the committed value with security parameter $k$, for large
enough $k$'s.

Next, we need to show that for every $k,\eps$ for which
$F_1^{(k)}(M_1,\eps)\not=\perp$ and $F_2^{(k)}(M_2,\eps)\not=\perp$
the profile $(F_1^{(k)}(M_1,\eps),F_2^{(k)}(M_2,\eps))$ constitutes
an $\eps$-TFNE in the
$T=(T^{(k)}_{1,\eps},T^{(k)}_{2,\eps})$-constrained version of
$\Gamma^{(k)}$. Let $k,\eps$ be as above, and let
$\sigma=(\sigma_1,\sigma_2)=(F_1^{(k)}(M_1,\eps),F_2^{(k)}(M_2,\eps))$.
We first show that $\sigma$ constitutes an $\eps$-NE in the
$T$-constrained version of
$\Gamma^{(k)}$.

The strategy $\sigma_1$ chooses a random commitment of a random bit in round 1,
and in round 3 decommits whenever it can. It is easy to see that
this is optimal, as player 2 always guesses the committed value with probability
$1/2$, and so there is no strategy for player 1 for which he can decommit with probability greater than $1/2$ in round 3. It is also easy to see that player 2's strategy is an $\eps$
best-response, as any PPT ITM $M_2$ for player 2 for which
$F_2^{(k)}(M_2,\eps)\not=\perp$ does not guess with advantage more
than $\eps$. We conclude that $\sigma$ constitutes an $\eps$-NE in
the $T$-constrained version of the game $\Gamma^{(k)}$.

Next, we show that no player is facing an $\eps$-threat with respect to
$\sigma$ at any round of the
$T$-constrained version of $\Gamma^{(k)}$. Note that for both players, the expected payoff
according to $\sigma$ is $1/2$. Suppose some player is facing an $\eps$-threat
with respect to $\sigma$. We divide the proof into
cases.
\paragraph{Case 1 -- $P_1$ is facing an $\eps$-threat in round 3:} In order for
$P_1$ to improve in Step 3 by more than $\eps$, it
must play a round 3 strategy $\tau(3)$ in which he sends ${\sf decom}$ that proves that $r_1 \neq r_2$ with larger
probability than in $\sigma$. However, since in $\sigma$ player 1 sends
${\sf decom}$ whenever $r_1 \neq r_2$ (and otherwise no such ${\sf
decom}$ exists, since the commitment is perfectly binding), we
conclude that no such $\tau(3)$ exists.
\paragraph{Case 2 -- $P_2$ is facing an $\eps$-threat in round 2:}
According to the constraints, $P_2$ cannot guess $r_1$ with
probability greater than $1/2+\eps$. So in order for him to improve by \textit{more} than $\eps$,
it must be the case that he has some round 2 strategy $\tau(2)$, such that in any $\eps$-threat-free continuation
in $\mathrm{Cont}(\sigma(1),\tau(2))$ player 1
aborts with positive probability conditioned on $r_1\not=r_2$. However,
any continuation where $P_1$ aborts with zero probability
conditioned on $r_1\not=r_2$ (and sends ${\sf decom}$) is $\eps$-threat-free,
and so there is no deviation for $P_2$ for which he improves on {\em
all} $\eps$-threat-free continuation.
\paragraph{Case 3 -- $P_1$ is facing an $\eps$-threat in round 1:}
Since $\sigma$ is $\eps$-threat-free on round 1, if $P_1$ is threatened in round 1 then
he has a round 1 strategy $\tau(1)$ so that for all $\eps$-threat-free
profiles in $\mathrm{Cont}(\tau(1))$ his expected payoff is greater than $1/2+\eps$.
Consider the profile $\sigma'=(\tau(1),\sigma(2),\sigma(3))$. This profile gives both players an expected
payoff of $1/2$ (assuming $\tau(1)$ aborts with probability 0, which is clearly optimal), and is $\eps$-threat-free on round 2 (by the same argument as Case 1 above).
If $\sigma'$ is $\eps$-threat-free on round 1 as well, then $P_1$ does not improve by more than $\eps$
using the deviation $\tau(1)$. If $\sigma'$ is not $\eps$-threat-free on round 1, then in any $\eps$-threat-free
profile in $\mathrm{Cont}(\tau(1))$ player 2's payoff must be greater than $1/2+\eps$. However, this means that
$P_1$'s payoff is less than $1/2$, and again he does not improve using the deviation $\tau(1)$.
Hence, the postulated $\tau(1)$ does not exist, and so $P_1$ is not facing an $\eps$-threat in round 1.
\end{proof}
\fi

\section{Correlated Equilibria Without a Mediator}
\label{sec-DHR}

In one of the first papers to consider the intersection of game
theory and cryptography, Dodis, Halevi and Rabin proposed an
appealing methodology for implementing a correlated equilibrium in a
2-player normal-form game without making use of a
mediator~\cite{DHR00}. Under standard hardness assumptions, they
showed that for any 2-player normal-form game $\Gamma$ and any
correlated equilibrium $\sigma$ for $\Gamma$, there exists a new
2-player extensive ``extended game" $\Gamma'$ and a CNE
$\sigma'$ for $\Gamma'$, such that $\sigma$ and $\sigma'$ achieve
the same payoffs for the players.
\ifnum\draft=2
\else
(Strictly speaking $\Gamma'$ is a
sequence of games indexed by a security parameter, and a CNE is
defined for a sequence.)
\fi
However, as already pointed out by Dodis et
al., their protocol lacks a satisfactory analysis of its sequential
nature -- the resulting ``extended game" is an extensive game,
but the solution concept they use, CNE, is not strong enough for
these games.

In the following, we extend the definition of CTFNE to allow
handling this setting (that is, we define CTFNE for extensive games
with simultaneous moves at the leaves), give some justification for
our new definition, and then provide a new protocol for removing the
mediator that achieves CTFNE in a wide class of correlated
equilibria that are in the convex hull of Nash equilibria (see
definition below).

\ifnum\draft=2
\else
\subsection{The Dodis-Halevi-Rabin Protocol}
\label{DHR.app}

The ``extended game" $\Gamma'$ consists of 2 phases. In the first
phase (``preamble phase"), the players execute a protocol for
sampling a pair under the distribution $\sigma$, and in the second
phase each player plays the action implied by the sampled pair, in
the original normal-form game. The CNE of the extended game is the
profile that consists of each player playing the protocol honestly
in the first phase, and then in the second phase, if the other
player did not abort, choosing the action by the protocol's
result, and otherwise ``punishing" the other player by choosing a
``min-max" action (i.e., choosing an action minimizing the utility
resulting from the other player's best response).

This profile is indeed a CNE because an efficient player can achieve
only a negligible advantage by trying to break the cryptography in
the first phase, cannot achieve any advantage by aborting in the
first phase (as this minimizes its best possible move in the second
phase), and cannot gain any advantage in the expectation of the
payoff by deviating in the second phase, because the players are
playing a pair of actions from a correlated equilibrium.

\fi

\ifnum\draft=2
\subsection{TFNE for Games with Simultaneous Moves at the Leaves}
\label{GSML.sec}

For a formal definition of extensive game with simultaneous moves see
Osborne and Rubinstein \cite{OR94}. In order to adjust our definition for extensive games with
simultaneous moves, we notice that when a player deviates on a
history with a simultaneous move, he cannot expect the other to
react to this deviation (because they both play
at the same time). However, in order to argue that a profile is
rational, we still need to require that for every simultaneous move
in the equilibrium support, each player is playing a ``best response"
given the other player's prescribed behavior. This means the prescribed behavior for the players should form some kind of equilibrium for normal-form games. In our case, the prescribed behavior will form a NE. The question of what
should a CTFNE profile prescribe in off-equilibrium-support histories is more delicate: Clearly, in order to claim that the profile is ``rational," again we need some kind of equilibrium for normal-form games. But in this case one can
argue that after one player deviated, the other player cannot assume the
deviating player will play his prescribed behavior in the
simultaneous move (as he is already not following his prescribed
behavior). However, we argue that it is in fact still rational to
assume the deviating player will play his prescribed behavior. The
justification for this claim is essentially the same as the
justification for the rationality of NE. Once there is a prescribed
behavior that is a NE, each player knows the other has no incentive
to deviate, and so he also has no incentive to deviate.

Thus, our new definition of TFNE for extensive games with
simultaneous moves at the leaves (abbreviated GSML), is essentially
the same as the original definition, except that (i) we require a profile in TFNE to prescribe a NE
in any terminal leaf, and (ii) in the definition of a threat we
do not allow a player to assume the other will deviate from his
strategy in any NE. In order to
formally modify our definition of TFNE to achieve (ii), essentially
we would need to define the only threat-free continuation on a leaf
to be the one that assigns to the players the actions in the
prescribed NE (which expresses the idea that a player is not allowed
to assume the other will deviate from his strategy in any NE).

However, we adopt an equivalent, simpler convention. Given a GSML
$\Gamma$ and a profile $\sigma$ that assigns a NE at every simultaneous
move, we look at a slightly modified game $\Gamma'$: All
simultaneous moves are removed, and instead at each leaf where a simultaneous move was removed each player is assigned
his expected
payoff in the corresponding NE for that leaf. Note that the modified game is now a regular extensive game
with \textit{no} simultaneous moves. We then ``prune" the
profile to remove all the distributions on actions on all
simultaneous leaves and denote the resulting profile $\sigma'$. We say
that Ê$\sigma$ is a TFNE in $\Gamma$ if $\sigma'$ is a TFNE in
$\Gamma'$. We call $\Gamma'$ and $\sigma'$ the {\em pruned
representation} of $\Gamma$ and $\sigma$.

The definition of CTFNE for GSML is derived from the above
definition of TFNE for GSML, similarly to the derivation of CTFNE from
TFNE in the non-simultaneous case.

It seems that
for general GSML's our definition is too strong, because
in certain cases it is computationally intractable to compute the assigned NE in every leaf. While we do not yet know how to relax our definition to apply to these
cases, we believe our definition, when met, is sufficient.

\else
\subsection{TFNE for Games with Simultaneous Moves at the Leaves}
\label{GSML.sec}

The definition of an extensive game with simultaneous moves is
similar to the definition of an ordinary extensive game. The
main difference is that now the function $P$ maps to (nonempty)
sets of players rather than to single players. The definition of
history is then changed  to a sequence of sets of actions
rather than a sequence of actions, and the definitions of a
strategy and a payoff function
are both also changed accordingly. For a formal definition see
Osborne and Rubinstein \cite{OR94}.

In order to adjust our definition for extensive games with
simultaneous moves, we notice that when a player deviates on a
history with a simultaneous move, he cannot expect the other to
react to this deviation (because they both play
at the same time). However, in order to argue that a profile is
rational, we still need to require that for every simultaneous move
in the equilibrium support, each player is playing a ``best response"
given the other player's prescribed behavior. This means the prescribed behavior for the players should form some kind of equilibrium for normal-form games. In our case, the prescribed behavior will form a NE. The question of what
should a CTFNE profile prescribe in off-equilibrium-support histories is more delicate: Clearly, in order to claim that the profile is ``rational," again we need some kind of equilibrium for normal-form games. In our case the only deviation will be prematurely aborting
without completing the preamble phase, which leads to the original
normal-form game without agreeing on a sampled pair. In this case one can
argue that after one player aborted, the other (non-aborting) player cannot assume the
aborting player will play his prescribed behavior in the
simultaneous move (as he is already not following his prescribed
behavior). However, we argue that it is in fact still rational to
assume the aborting player will play his prescribed behavior. The
justification for this claim is essentially the same as the
justification for the rationality of NE. Once there is a prescribed
behavior that is a NE, each player knows the other has no incentive
to deviate, and so he also has no incentive to deviate. The
essential difference between a deviation in an extensive game and a
deviation in a simultaneous move, is that in the former, once a
player deviated, the other player is facing a fact. He now has to
readjust his behavior according to this deviation. However, in the
latter, there is no point for a player to deviate from the prescribed
NE, because the other player will not know about this deviation
prior to choosing his move (if at all).  Thus, for terminal leaves that are off-equilibrium-support (i.e., in the original normal-form game that follows an abort of some player), we claim it is sufficient for a CTFNE to prescribe a NE as well.

The bottom line of this discussion is that players cannot assume
other players will deviate from any prescribed NE in any terminal
leaf. Thus, our new definition of TFNE for extensive games with
simultaneous moves at the leaves (abbreviated GSML) is essentially
the same as the original definition, except that (i) we require a profile in TFNE to prescribe a NE
in any terminal leaf, and (ii) in the definition of a threat we
do not allow a player to assume the other will deviate from his
strategy in any NE at a terminal leaf. In order to
formally modify our definition of TFNE to achieve (ii), essentially
we would need to define the only threat-free continuation on a leaf
to be the one that assigns to the players the actions in the
prescribed NE (which expresses the idea that a player is not allowed
to assume the other will deviate from his strategy in any NE).

However, we adopt an equivalent, simpler convention. Given a GSML
$\Gamma$ and a profile $\sigma$ that assigns a NE at every simultaneous
move, we look at a slightly modified game $\Gamma'$: All
simultaneous moves are removed, and instead at each leaf where a simultaneous move was removed each player is assigned
his expected
payoff in the corresponding NE for that leaf. Note that the modified game is now a regular extensive game
with \textit{no} simultaneous moves. We then ``prune" the strategy
profile to remove all the distributions on actions on all
simultaneous leaves and denote the resulting profile $\sigma'$. We say
that  $\sigma$ is a TFNE in $\Gamma$ if $\sigma'$ is a TFNE in
$\Gamma'$. We call $\Gamma'$ and $\sigma'$ the {\em pruned
representation} of $\Gamma$ and $\sigma$.

The definition of CTFNE for GSML is derived from the above
definition of TFNE for GSML, similarly to the derivation of CTFNE from
TFNE in the non-simultaneous case.

\paragraph{A note on the strength of our definition} It seems that
for general GSMLs our definition is too strong. The reason is that
in certain cases it is computationally intractable for the players
to play the prescribed NE in every leaf (it is easy to construct
simple sequences of games where one cannot assign tractable Nash equilibria at all leaves).
While we do not yet know how to relax our definition to apply to these
cases, we believe our definition, when met, is sufficient.

\subsection{Our Protocol}
For a non-trivial class of correlated equilibria, we show how to
modify the DHR protocol to achieve CTFNE. Our basic idea is to use
Nash equilibria as ``punishments" for aborting players. That is, if
there is a NE that assigns to a player a payoff at most his expected
payoff when not aborting, then assigning this NE in case he aborts
serves as a punishment and yields that the player has no incentive
to abort. In the following we characterize a family of correlated
equilibria for which we can use the aforementioned punishing
technique, and prove that for this family we can remove the mediator
while achieving CTFNE.

We say that a correlated equilibrium $\pi$ is a {\em convex
combination of Nash equilibria} if $\pi$ is induced by a distribution on (possibly mixed) Nash equilibria. (The set of such distributions is
sometimes referred to as the {\em convex hull of Nash equilibria}.)
Note that any such distribution is a correlated equilibrium (CE), but the
converse is not true.

Let $\pi$ be a correlated equilibrium for a two-player game
$\Gamma$ that is a convex combination of a set $N$ of NEs. We say
that $\pi$ is {\em weakly Pareto optimal} if there does not exist a different CE
$\rho$ in the convex hull of $N$ for which both
$\E[u_1(O(\rho))]> \E[u_1(O(\pi))]$ and $\E[u_2(O(\rho))]> \E[u_2(O(\pi))]$.

We say that a distribution is {\em samplable} if there exists a probabilistic TM that halts on every infinite randomness vector, and can sample it. This is equivalent to requiring that all probabilities can be expressed in binary (assuming we work over $\zo$). Note that every distribution can be approximated arbitrarily accurately by a samplabale distribution.

\begin{procthm} Assume there exists a non-interactive computationally binding commitment scheme. Let $\pi$ be a weakly Pareto optimal correlated equilibrium for a two-player game $\Gamma$ that is a samplable convex combination $\Pi$ of some set of samplable Nash equilibria. Then there exists an extended extensive game and a
profile that achieves the same expected payoffs as $\pi$ and is a
CTFNE.
\end{procthm}

\begin{proof}
Since $\Pi$ is samplable, the common denominator of all probabilities in $\Pi$ is a power of two. Thus, we can assume $\Pi$ is a {\em uniform} distribution on a sequence of Nash equilibria that may contain repetitions, where the length of the sequence is a power of two. Let $2^\ell$ be the length of that sequence, and let $(\pi_{0^\ell},\dots,\pi_{1^\ell})$ be that sequence. Note that the distribution $\pi$ can now be generated by first choosing uniformly at random a string $r$ in $\zo^\ell$, and then choosing a pair of actions according to $\pi_r$.

Let $\widehat{\sigma}^i$ be the NE that
assigns the worst payoff for $P_i$ (this value represents the
``severest punishment" for player $i$).

Our protocol embeds a 2-party string sampling protocol, which is a simple generalization of the Blum coin flipping protocol \cite{Bl}. The protocol consists of simply running the Blum protocol in parallel for a fixed number of times. This protocol, in turn, relies on a perfectly binding commitment scheme as in Section $\ref{app:proof-chicken}$, whose formal definition can be found in Appendix~\ref{apx:defs}.

\ifnum\draft=2
We
\else
As in Section $\ref{app:proof-chicken}$, we describe the two ITMs that form the protocol in an interleaved manner. We
\fi
denote the ITMs playing the strategies of
$P_1,P_2$ by $M_1,M_2$, respectively.
\begin{itemize}
\item {\sf Round 1:} Player 1 chooses uniformly at random a string $r=(r_1,\dots,r_\ell)$ from $\zo^\ell$, and sends $c=(c_1={\sf com}^{(k)}(r_1),\dots,c_\ell={\sf com}^{(k)}(r_\ell))$ to
player 2 (player 1 also obtains $({\sf decom}_1,\dots,{\sf decom}_\ell)$, where ${\sf decom}_i$ is a legal decommitment with respect to $c_i$ and $r_i$).
\item {\sf Round 2:} If player 1 aborted, the assigned NE is $\widehat{\sigma}^1$. Else, player 2 chooses a uniformly random string $r'=(r'_1,\dots,r'_\ell)$ from $\zo^\ell$, and sends $r'$ to player
1.
\item {\sf Round 3:} If player 2 aborted, the assigned NE is $\widehat{\sigma}^2$. Else, player 1 sends the message $((r_1,{\sf decom}_1),\dots,(r_\ell,{\sf decom}_\ell))$.
\item If player 1 aborted, the assigned NE is $\widehat{\sigma}^1$. Else, player 2 verifies that ${\sf decom}_i$ is a legal decommitment with respect to $c_i$ and $r_i$ for $1 \leq i \leq \ell$. If the verification fails (which is equivalent to an abort of player 1, as it means player 1 sent an illegal message), the assigned NE is $\widehat{\sigma}^1$. Else, the assigned NE is $\pi_{r \oplus r'}$ (where $\oplus$ is bitwise exclusive-or).
\end{itemize}
\ifnum\draft=2
We now show that the pair $(M_1,M_2)$ forms a CTFNE for the protocol above.
\else
\begin{lemma}
The pair $(M_1,M_2)$ forms a CTFNE for the protocol above.
\end{lemma}
\begin{proof}
\fi
Let $\{\tilde{\Gamma}^{(k)}\}_{k \in \mathbb{N}}$ be the sequence of games induced by the protocol. Denote the pruned
representation of $\tilde{\Gamma}^{(k)}$ by $\Gamma^{(k)}$. Let
$\tilde{\sigma}^{(k)}_1,\tilde{\sigma}^{(k)}_2$ be the strategies of
$P_1,P_2$ in the protocol with security parameter $k$,
and let $\sigma^{(k)}_1,\sigma^{(k)}_2$ be their pruned
representations. Let $\sigma^{(k)} =
(\sigma^{(k)}_1,\sigma^{(k)}_2)$. We prove that $\{\sigma^{(k)}\}$
is a CTFNE in $\{\Gamma^{(k)}\}$, which, by the discussion of Section~\ref{GSML.sec},
implies that $\{\tilde{\sigma}^{(k)}\}$ is a CTFNE in
$\{\tilde{\Gamma}^{(k)}\}$.

First we define the functions $F_1^{(k)}$ and $F_2^{(k)}$. For any $k$,
the function $F_1^{(k)}$ never maps to $\bot$ (this, roughly
speaking, reflects the fact that the protocol is secure against an
all-powerful player 1, which follows from the perfect binding property of the commitment scheme). For $F_2$ we use the following rule:
$F_2^{(k)}(M,\eps)=\bot$ if and only if
\be\label{rule-F2}\E[u_2^{(k)}(O(\sigma_1^{(k)},\sigma^{(k)}_M))] \geq
\E[u_2^{(k)}(O(\sigma^{(k)}))]+\eps,\ee
where $\sigma^{(k)}_M$ is the
strategic representation of machine $M$ and $\sigma_1^{(k)}$ is the
strategic representation of machine $M_1$, both with security parameter
$k$. In other words, $P_2$ cannot unilaterally $\eps$-improve in the $(T^{(k)}_{1,\eps},T^{(k)}_{2,\eps})$-constrained version of
$\Gamma^{(k)}$.

The fact that $F_1$ is PPT-covering is straightforward. The fact
that $F_2$ is PPT covering follows from the security of the
commitment scheme, as we prove next.

\begin{claim}
The strategy-filter $F_2$ is PPT-covering.
\end{claim}

\begin{itemize}
\item[] \begin{proof}
Suppose $F_2$ is not PPT-covering. Then from (\ref{rule-F2}) there is a PPT ITM $M$ and a
polynomial $p$ such that \be
\E[u_2^{(k)}(O(\sigma^{(k)}_1,\sigma_M^{(k)}))]\geq\E[u_2^{(k)}(O(\sigma^{(k)}_1,\sigma^{(k)}_2))]+1/p(k)
\label{equ-adv}\ee for infinitely many $k$'s, where $\sigma_M^{(k)}$ is
the strategic representation of the machine $M$ with security
parameter $k$.

First, we show that we can assume $M$ does not abort in round 2. An abort of $P_2$ leads to a leaf with $\widehat{\sigma}^2$.
But since $\pi$ is a convex combination of NEs, following the
protocol would mean playing a NE. Since by definition
$\widehat{\sigma}^2$ is the worst NE for player 2, it follows that
the machine $M'$ that behaves the same as $M$, but whenever $M$
aborts, $M'$ instead follows the protocol (i.e.\ acts like $M_2$) does at least as well as $M$. The machine $M'$ is well-defined, as the reduced strategy $\sigma^{(k)}_2$ is in fact a full strategy, and is defined everywhere.\footnote{Note that we assume here that there exists such PPT ITM $M'$. This may not always be the case. One reason is that sometimes detecting with probability $1$ whether $M$ aborted cannot be done in polynomial time (or at all). The reason is that any illegal message is regarded as abort, but sometimes a party cannot ``know" whether its message is illegal or not. See \cite{AL07}, Section 6.3 for an example. Another reason could be that in order to ``emulate" $M_2$, the machine $M'$ needs to be in some internal state. We note, however, that in our case neither problem occurs.}

Since the payoffs in $\{\Gamma^{(k)}\}$ are bounded in $k$ and the number of NEs in $\pi$ is fixed in $k$, by (\ref{equ-adv}) there exists a polynomial $p$ and (at least one) $s \in \zo^\ell$ such that for infinitely many $k$'s
$$\Pr[O(\sigma^{(k)}_1,\sigma_M^{(k)})=\pi_s]-\Pr[O(\sigma^{(k)}_1,\sigma_2^{(k)})=\pi_s]\geq 1/{p(k)}.$$

It follows that for infinitely many $k$'s
\be\label{r-r'}\Pr_{(\sigma^{(k)}_1,\sigma_M^{(k)})}[ r \oplus r' = s]\ - \Pr_{ (\sigma^{(k)}_1,\sigma_2^{(k)})}[ r \oplus r' = s] \geq 1/{p(k)}.\ee

\begin{claim}
There exists a polynomial $q$ such that for each $k$ satisfying (\ref{r-r'}) there exists some $i\in\{1,\ldots,\ell\}$ for which
\be \label{adv-m}\Pr_{(\sigma^{(k)}_1,\sigma_M^{(k)})}[r_i \oplus r'_i=s_i | r_j \oplus r'_j=s_j\mbox{ }\forall j<i]-1/2\geq 1/{q(k)}.\ee
\end{claim}

\begin{itemize}
\item[]
\begin{proof}
We show that the claim holds with  $q(k)=2^\ell\cdot p(k)$. Let $k$ be such that (\ref{r-r'}) holds, and suppose towards a contradiction
that (\ref{adv-m})  does not hold for any $i\in\{1,\ldots,\ell\}$. Then
\begin{align*}
\Pr_{(\sigma^{(k)}_1,\sigma_M^{(k)})}[ r \oplus r' &= s]\ - \Pr_{(\sigma^{(k)}_1,\sigma_2^{(k)})}[ r \oplus r' = s] \\
= &\Pr_{(\sigma^{(k)}_1,\sigma_M^{(k)})}[r_1 \oplus r'_1=s_1]\cdot\Pr_{(\sigma^{(k)}_1,\sigma_M^{(k)})}[r_2 \oplus r'_2=s_2|r_1 \oplus r'_1=s_1]\cdot\ldots\\  &\cdot\Pr_{(\sigma^{(k)}_1,\sigma_M^{(k)})}[r_\ell \oplus r'_\ell=s_\ell | r_j \oplus r'_j=s_j\mbox{ }\forall j<\ell]
 -\Pr_{(\sigma^{(k)}_1,\sigma_2^{(k)})}[ r \oplus r' = s]\\
< &\left(\frac{1}{2}+\frac{1}{q(k)}\right)^\ell - \frac{1}{2^\ell}\\
< &\frac{2^\ell}{q(k)} = \frac{1}{p(k)}.
\end{align*}

The first inequality holds since the distribution on $r \oplus r'$ in $(\sigma^{(k)}_1,\sigma_2^{(k)})$ is uniform on $\{0,1\}^\ell$.
The second inequality follows from the observation that in $\left(1/2+1/{q(k)}\right)^\ell$ we are summing over $2^\ell$ terms, one equal to $1/2^\ell$ and the others strictly smaller than $1/q(k)$. Thus, we get a contradiction to (\ref{r-r'}).
\end{proof}
\end{itemize}

Since there are infinitely many $k$'s for which (\ref{adv-m}) holds, and because $\ell$ is fixed, there must exist some $i\in\{1,\ldots,\ell\}$ for which (\ref{adv-m}) holds infinitely often.
This, however, yields a PPT machine $A$ that breaks the hiding property of the commitment scheme: Given a commitment $c={\sf com}^{(k)}(r)$ for a uniformly chosen random bit $r$, the machine $A$ chooses uniformly at random a string $(r_1,\dots,r_{i-1},r_{i+1},\ldots, r_\ell)$ from $\zo^{\ell-1}$, and runs $M$ on
$$(c_1={\sf com}^{(k)}(r_1),\dots,c_{i-1}={\sf com}^{(k)}(r_{i-1}),c,c_{i+1}={\sf com}^{(k)}(r_{i+1}),\ldots, c_\ell={\sf com}^{(k)}(r_\ell))$$
to get output $r'$. Then, if $r_j \oplus r'_j=s_j\mbox{ }\forall j<i$, algorithm $A$ outputs $s_i \oplus r'_i$, and otherwise $A$ outputs a uniformly random bit.
Clearly $A$ is a PPT machine. From (\ref{r-r'}) it follows that
infinitely often, with probability at least $1/2^\ell$ it will be the case that
$r_j \oplus r'_j=s_j\mbox{ }\forall j<i$.
Once $r'$ is such that $r_j \oplus r'_j=s_j\mbox{ }\forall j<i$,
(\ref{adv-m}) implies that $\Pr[s_i \oplus r'_i=r_i | r_j \oplus r'_j=s_j\mbox{ }\forall j<i]\geq 1/2 + q(k)$. Thus, in total, for infinitely many
$k$'s it holds that 
$$\Pr[s_i \oplus r'_i=r_i] =\left (1-\frac{1}{2^{\ell}}\right)\cdot \frac{1}{2} + \frac{1}{2^{\ell}}\cdot\left(\frac{1}{2}+ q(k)\right)=\frac{1}{2}+\frac{q(k)}{2^{\ell}},$$
which means that $A$ breaks the hiding property of the commitment scheme. This is a contradiction.
\end{proof}
\end{itemize}

Next, we show that for all $k,\eps$ for which
$F_1^{(k)}(M_1,\eps)\not=\perp$ and $F_2^{(k)}(M_2,\eps)\not=\perp$
the profile $(F_1^{(k)}(M_1,\eps), F_2^{(k)}(M_2,\eps))$ constitutes
an $\eps$-TFNE in the
$T=(T^{(k)}_{1,\eps}, T^{(k)}_{2,\eps})$-constrained version of
$\Gamma^{(k)}$. Let $k,\eps$ be as above, and let
$\sigma=(\sigma_1,\sigma_2)=(F_1^{(k)}(M_1,\eps),F_2^{(k)}(M_2,\eps))$.
\ifnum\draft=2
\else
We first show that $\sigma$ constitutes an $\eps$-NE in the
$T$-constrained version of
$\Gamma^{(k)}$.
\fi
Suppose $P_1$ unilaterally $\eps$-improves in the $T$-constrained version of $\Gamma^{(k)}$.
From similar arguments as above we can assume $P_1$ never aborts. But when $P_1$ never aborts the outcome is exactly $\pi$, as the players are playing $\pi_{r \oplus r'}$, and $r'$ is chosen uniformly at random.
\ifnum\draft=2
\else

\fi
Suppose now that $P_2$ unilaterally $\eps$-improves in the $T$-constrained version of $\Gamma^{(k)}$. However, this is a contradiction to the constraints, that state that for any $k$ $P_2$ cannot unilaterally $\eps$-improve in the $(T^{(k)}_{1,\eps},T^{(k)}_{2,\eps})$-constrained version of
$\Gamma^{(k)}$.

Next, we show that no player is $\eps$-threatened with respect to
$\sigma$ at any round of the
$T$-constrained version of  $\Gamma^{(k)}$. To this end, suppose towards a contradiction that some player is $\eps$-threatened
with respect to $\sigma$. We divide the proof into
cases.

\ifnum\draft=2
{\bf Case 1 -- $P_1$ is facing an $\eps$-threat in round 3:}
\else
\paragraph{Case 1 -- $P_1$ is facing an $\eps$-threat in round 3:}
\fi
In step 3 player 1 has exactly two options: He can (i) play honestly,
send $((r_1,{\sf decom}_1),\dots,(r_\ell,{\sf decom}_\ell))$ which he generated in round 1, and receive $\E[u_1(O(\sigma))]$, or he can (ii) abort and receive $\E[u_1(O(\widehat{\sigma}^1))]$. The value $\E[u_1(O(\widehat{\sigma}^1))]$ is at most $\E[u_1(O(\sigma))]$, and so $P_1$ cannot improve over $\E[u_1(O(\sigma))]$. Hence player 1 is not facing an $\eps$-threat at round 3.

\ifnum\draft=2
{\bf Case 2 -- $P_2$ is facing an $\eps$-threat in round 2:}
\else
\paragraph{Case 2 -- $P_2$ is facing an $\eps$-threat in round 2:}
\fi
We first note that for any round 1 strategy for $P_1$ and round 2 strategy for $P_2$, the round strategy of playing honestly in round 3 for $P_1$ is threat-free, since he cannot improve over that strategy (again, since his only deviation is aborting, which gives him the worst possible NE). Thus, if $P_2$ is $\eps$-threatened at round 2, he has some round strategy that $\eps$-improves over $\E[u_2(O(\sigma))]$ when $P_1$ plays in round 3 (and 1) according to the protocol. This means that $P_2$ unilaterally $\eps$-improves, which contradicts the constraints (as well as the $\eps$-NE).

\ifnum\draft=2
{\bf Case 3 -- $P_1$ is facing an $\eps$-threat in round 1:}
\else
\paragraph{Case 3 -- $P_1$ is facing an $\eps$-threat in round 1:}
\fi
If $P_1$ is $\eps$-threatened in round 1, he has some round 1 strategy $\tau(1)$ for which every $\eps$-threat-free continuation $\eps$-improves over every $\eps$-threat-free continuation of $\sigma_1(1)$. We will  describe an $\eps$-threat-free continuation of $\tau(1)$ and an
$\eps$-threat-free continuation of $\sigma_1(1)$ that contradict this.

{ The $\eps$-threat-free continuation of $\sigma_1(1)$}: We established in Case 2 that when $P_1$ plays honestly in round 1, if $P_2$ plays honestly in round 2 he is not $\eps$-threatened. We also established there that $P_1$ playing honestly in round 3 is always $\eps$-threat-free. If follows that the continuation of both players playing honestly in rounds 2 and 3 is an $\eps$-threat-free continuation of $\sigma_1(1)$. On this profile $P_1$ receives $\E[u_1(O(\sigma))]$.

{ The $\eps$-threat-free continuation of $\tau_1(1)$}: As we established in Case 2, playing honestly in round 3 is always $\eps$-threat-free for $P_1$. Now, note that there is no profile in which both players improve simultaneously -- because all leaves are Nash equilibria, such a profile would be a distribution on Nash equilibria that contradicts the Pareto-optimality of $\pi$. Note also that because $P_1$ receives the worst possible payoff when he aborts, it follows that he improves also conditioned on not aborting (as this can only help him). Thus, in any threat-free continuation of $\tau(1)$, conditioned on $P_1$ not aborting in round 1, $P_2$ again cannot improve over $\E[u_2(O(\sigma))]$, as this again contradicts the Pareto-optimality of $\pi$. However, if $P_2$ plays honestly in round 2 and then $P_1$ plays honestly in round 3, then $P_2$ receives exactly $\E[u_2(O(\sigma))]$ conditioned on $P_1$ not aborting in round 1. It follows that this continuation is the best possible for $P_2$, and thus $P_2$ is not $\eps$-threatened in round 2 of this continuation. It follows that this continuation is $\eps$-threat-free. However, in this continuation $P_1$ receives $\E[u_1(O(\sigma))]$ conditioned on not aborting, and thus receives at most $\E[u_1(O(\sigma))]$ without the conditioning.
\end{proof}
\ifnum\draft=2
\else
This completes the proof of the theorem.
\end{proof}
\fi

\section{A General Theorem}
In this section we prove a general theorem identifying sufficient conditions for a strategy profile
to be a TFNE. The first condition is that the profile must be weakly Pareto optimal:

\begin{procdef}{Weakly Pareto optimal}
A strategy profile $\sigma\in T$ of an extensive game $\Gamma=(H,P,A,u)$ with constraints $T$
is {\em weakly Pareto optimal} if there does not exist
a strategy profile $\pi\in T$ for which both
$\E[u_1(O(\pi))]> \E[u_1(O(\sigma))]$ and $\E[u_2(O(\pi))]> \E[u_2(O(\sigma))]$.
\end{procdef}

Next, we require the profile to be $\eps$-safe. Intuitively, this just means that a player cannot harm the other too much
by a unilateral deviation (as opposed to not being able to gain too much, which is the NE condition).

\begin{procdef}{$\eps$-safe}
A strategy profile $\sigma=(\sigma_1,\sigma_{2})\in T$ of an extensive game $\Gamma=(H,P,A,u)$ with constraints $T=(T_1,T_2)$
is {\em $\eps$-safe} if for each player $i$,
$$\E\left[u_{-i}\left(O(\sigma)\right)\right] \geq \E\left[u_{-i}\left(O(\sigma'_i,\sigma_{-i})\right)\right]-\eps$$ for
every strategy $\sigma'_i\in T_i$ of player $i$.
\end{procdef}

Finally, we have the following theorem. Note that we are implicitly assuming that the extensive games in the claim are derived
from a cryptographic protocol or some other setting in which it is natural to discuss the ``rounds'' of a game.

\begin{procthm}
\label{thm:general}
Let $\Gamma=(H,P,A,u)$ be an extensive game with constraints
$T=(T_1,T_2)$, and let $\sigma=(\sigma_1,\sigma_2)$ be a weakly Pareto optimal $\eps$-NE of $\Gamma$
that is $\eps$-safe. Then $\sigma$ is an $\eps$-TFNE of $\Gamma$.
\end{procthm}

We also have the following corollary.

\begin{corollary}{\em
Let $\Gamma=(H,P,A,u)$ be a {\em zero-sum} extensive game with constraints
$T=(T_1,T_2)$, and let $\sigma$ be an $\eps$-NE of $\Gamma$.
Then $\sigma$ is an $\eps$-TFNE of $\Gamma$.}
\end{corollary}

The corollary follows from the observation that any $\eps$-NE of a zero-sum game is both
weakly Pareto optimal and $\eps$-safe. Note that the corollary implies the threat-freeness
part of Theorem~\ref{thm:coin-flipping-tfne}.

We now prove Theorem~\ref{thm:general}.

\begin{proof}
Suppose towards contradiction that at least one of the players is facing an $\eps$-threat
with respect to $\sigma$ at some round. Let $R$ be the latest such round: that is, player $i$ is facing an $\eps$-threat at round $R$ with respect to $\sigma$, and  no player is facing an $\eps$-threat at any round $R'$ that follows $R$.

By Definition~\ref{def:threat} it follows that there exists a round $R$ strategy $\tau=\tau(R)$ for player $i$ such that
the set $\mathrm{Cont}(\sigma(1,\ldots,R\!-\!1),\tau(R))$ is nonempty, and such that for all
$\pi\in \mathrm{Cont}(\sigma(1,\ldots,R\!-\!1),\tau(R))$ and $\pi'\in \mathrm{Cont}(\sigma(1,\ldots,R))$
that are $\eps$-threat-free on $R$ it holds that
\be\E\left[u_{i}\left(O(\pi)\right)\right] >
\E\left[u_{i}\left(O(\pi')\right)\right]+\eps,\label{eqn:dev-gain2}\ee
where
\begin{align*}\sigma(1,\ldots,S)\eqdef \sigma(1),\ldots,\sigma(S)
\end{align*}
and
\begin{align*}\mathrm{Cont}(\sigma(1,\ldots,R))\eqdef\Big\{\pi\in T: \pi(S) = \sigma(S) \mbox{ for all } S\leq R\Big\}.
\end{align*}

Note that $\sigma\in \mathrm{Cont}(\sigma(1,\ldots,R))$. Also note that, because $R$ is the latest round on which an $\eps$-threat occurs, the profile $\sigma$ is $\eps$-threat-free on $R$.

Using inequality~(\ref{eqn:dev-gain2}) we can then infer that for any $\pi\in \mathrm{Cont}(\sigma(1,\ldots,R\!-\!1),\tau(R))$ that is $\eps$-threat-free on $R$ it holds that
\be\E\left[u_{i}\left(O(\pi)\right)\right] >
\E\left[u_{i}\left(O(\sigma)\right)\right]+\eps.\label{eqn:dev-gain3}\ee

Let $\pi^1\in \mathrm{Cont}(\sigma(1,\ldots,R\!-\!1),\tau(R))$ be one such $\eps$-threat-free profile, and let ${\sigma^{1}}=(\pi^1_i, \sigma_{-i})$.

Fix $R^1=R$ and $\tau^1=\tau$ for consistent notation.
We next ask, is player $i$ facing an $\eps$-threat with respect to ${\sigma^{1}}$
at any round $R'$ that follows $R^1$? If yes, let $R^2$ be the next such round: there is no $R'$ between $R^1$ and $R^2$
on which player $i$ is facing an $\eps$-threat with respect to ${\sigma^{1}}$.
By Definition~\ref{def:threat} it follows that there exists a round $R^2$ strategy $\tau^2$ for player $i$
such that $\mathrm{Cont}(\sigma^1(1,\ldots,R^2\!-\!1),\tau^2(R^2))$ is nonempty, and such that
for all $\pi\in \mathrm{Cont}(\sigma^1(1,\ldots,R^2\!-\!1),\tau^2(R^2))$ and $\pi'\in \mathrm{Cont}(\sigma^1(1,\ldots,R^2))$ that are $\eps$-threat-free on $R^2$ it holds that
$$\E\left[u_{i}\left(O(\pi)\right)\right] >
\E\left[u_{i}\left(O(\pi')\right)\right]+\eps.$$

Assume $\tau^2$ is maximal, in the sense that for any $\pi\in\mathrm{Cont}(\sigma^1(1,\ldots,R^2-1),\tau^2(R^2))$ that is $\eps$-threat-free on $R^2$, player $i$ is \textit{not} facing an $\eps$-threat at round $R^2$ with respect to $\pi$. Pick some arbitrary $\pi^{2}\in \mathrm{Cont}(\sigma^1(1,\ldots,R^2\!-\!1),\tau^2(R^2))$, and fix ${\sigma^{2}}=(\pi^2_i,\sigma_{-i})$.

We now repeat the above procedure, finding the next threat to player $i$ and letting him act on that threat, as follows.
For $t=3,4,\ldots$ we ask, is player $i$ facing an $\eps$-threat with respect to ${\sigma^{t-1}}$
at any round $R'$ that follows $R^{t-1}$? If yes, let $R^t$ be the next such round: there is no $R'$ between $R^{t-1}$ and $R^t$
on which player $i$ is facing an $\eps$-threat with respect to ${\sigma^{t-1}}$.

By Definition~\ref{def:threat} it follows that there exists a round $R^t$ strategy $\tau^t$ for player $i$
such that $\mathrm{Cont}(\sigma^{t-1}(1,\ldots,R^t\!-\!1),\tau^t(R^t))$ is nonempty, and such that for all
$\pi\in\mathrm{Cont}(\sigma^{t-1}(1,\ldots,R^t\!-\!1),\tau^t(R^t))$
and $\pi'\in\mathrm{Cont}(\sigma^{t-1}(1,\ldots,R^t))$ that are $\eps$-threat-free on $R^t$ it holds that
$$\E\left[u_{i}\left(O(\pi)\right)\right] >
\E\left[u_{i}\left(O(\pi')\right)\right]+\eps.$$

Assume $\tau^t$ is maximal, in the sense that for any $\pi\in\mathrm{Cont}(\sigma^{t-1}(1,\ldots,R^t\!-~\!1),\tau^t(R^t))$ that is $\eps$-threat-free on $R^t$, player $i$ is \textit{not} facing an $\eps$-threat at round $R^t$ with respect to $\pi$. Pick some arbitrary $\pi^{t}\in \mathrm{Cont}(\sigma^{t-1}(1,\ldots,R^t\!-\!1),\tau^t(R^t))$, and fix ${\sigma^{t}}=(\pi^t_i,\sigma_{-i})$.

Finally, after repeating this for all $t$ until there are no more $\eps$-threats to $P_i$ on any round that follows $R$,
we are left with a profile ${\sigma^{C}}=(\pi^C_i,\sigma_{-i})$ on which player $i$ is not facing an $\eps$-threat
at any round below $R$.

Fix $\rho = {\sigma^{C}}$, and recall that, by construction, $\rho_{-i}=\sigma_{-i}$. Because $\sigma$ is $\eps$-safe, it must be the case that
\be\E\left[u_{-i}\left(O(\rho)\right)\right] \geq
\E\left[u_{-i}\left(O(\sigma)\right)\right]-\eps.\label{eqn:constraints22}\ee

We next ask, is player $-i$ facing an $\eps$-threat with respect to $\rho$ at any round $S$ that follows $R$?
As the following claim shows, the answer is positive:
\begin{claim}
Player $-i$ is facing an $\eps$-threat with respect to $\rho$ at some round $S$ that follows $R$.
\end{claim}
\ifnum\draft=2
\else
\begin{itemize}
\item[]
\fi
\begin{proof}
Suppose not. By our construction of $\rho$, player $i$ is also not facing
an $\eps$-threat with respect to $\rho$ at any round that follows $R$. This means that the
profile $\rho$ is $\eps$-threat-free on the subgames $R$.

Since $\rho\in \mathrm{Cont}(\sigma(1,\ldots,R-1),\tau(R))$ and since  $\sigma\in \mathrm{Cont}(\sigma(1,\ldots,R))$ is $\eps$-threat-free on $R$, we can then use (\ref{eqn:dev-gain3}) to infer that
$$\E\left[u_{i}\left(O(\rho)\right)\right] >
\E\left[u_{i}\left(O(\sigma)\right)\right]+\eps.$$
However, since $\rho=(\pi^C_i,\sigma_{-i})$ is a \textit{unilateral} deviation of player $i$, this
contradicts the fact that $\sigma$ constitutes an $\eps$-NE.
\end{proof}
\ifnum\draft=2
\else
\end{itemize}
\fi

Let $S^1$ be the latest round on which $P_{-i}$ is facing an $\eps$-threat with respect to $\rho$.
By Definition~\ref{def:threat} it follows that there exists a round $S^1$ strategy $\mu^1$ for player $-i$
such that $\mathrm{Cont}(\rho(1,\ldots,S^1\!-\!1),\mu^1(S^1))$ is nonempty, and such that
for all $\pi\in\mathrm{Cont}(\rho(1,\ldots,S^1\!-\!1),\mu^1(S^1))$ and $\pi'\in \mathrm{Cont}(\rho(1,\ldots,S^1))$ that are $\eps$-threat-free on $S^1$ it holds that
$$\E\left[u_{i}\left(O(\pi)\right)\right] >
\E\left[u_{i}\left(O(\pi')\right)\right]+\eps.$$

Assume $\mu^1$ is maximal, in the sense that for any $\pi\in\mathrm{Cont}(\rho(1,\ldots,S^1\!-\!1),\mu^1(S^1))$  that is $\eps$-threat-free on $S^1$, player $-i$ is \textit{not} facing an $\eps$-threat at round $S^1$ with respect to $\pi$. Pick some  $\rho^{1}\in\mathrm{Cont}(\rho(1,\ldots,S^1\!-\!1),\mu^1(S^1))$ that is $\eps$-threat-free on $S^1$ -- such a $\rho^1$ must exist by Proposition~\ref{prop:well-defined}.

Now, note that because $S^1$ was the last round on which $P_{-i}$ is facing an $\eps$-threat, and because $P_i$ is not facing an $\eps$-threat at any
round following $R$ with respect to $\rho$, it must be the case that $\rho$ is $\eps$-threat-free on $S^1$. Since $\rho \in \mathrm{Cont}(\rho(1,\ldots,S^1))$ we then have that
$$\E\left[u_{-i}\left(O(\rho^1)\right)\right] >\E\left[u_{-i}\left(O(\rho)\right)\right] + \eps \geq
\E\left[u_{-i}\left(O(\sigma)\right)\right],$$
where the second inequality follows from~\eqref{eqn:constraints22}. We now repeat the above procedure, finding the preceding threat to player $-i$ (but that still follows $R$) and letting him act on that threat, as follows.
For $t=2,3,\ldots$ we ask, is $P_{-i}$ facing an $\eps$-threat with respect to $\rho^{t-1}$ at any round $S$ that follows $R$? If yes,
let $S^t$ be the latest such round.
By Definition~\ref{def:threat} it follows that there exists a round $S^t$ strategy $\mu^t$ for player $-i$
such that $\mathrm{Cont}(\rho^{t-1}(1,\ldots,S^t\!-\!1),\mu^t(S^t))$  is nonempty, and such that
for all $\pi\in\mathrm{Cont}(\rho^{t-1}(1,\ldots,S^t\!-\!1),\mu^t(S^t))$ and $\pi'\in \mathrm{Cont}(\rho^{t-1}(1,\ldots,S^t))$ that are $\eps$-threat-free on $S^t$ it holds that
$$\E\left[u_{i}\left(O(\pi)\right)\right] >
\E\left[u_{i}\left(O(\pi')\right)\right]+\eps.$$

Assume $\mu^t$ is maximal, in the sense that for any $\pi\in\mathrm{Cont}(\rho^{t-1}(1,\ldots,S^t\!-~\!1),\mu^t(S^t))$ that is $\eps$-threat-free on $S^t$, player $-i$ is \textit{not} facing an $\eps$-threat at round $S^t$ with respect to $\pi$. Pick some $\rho^{t}\in \mathrm{Cont}(\rho^{t-1}(1,\ldots,S^t\!-~\!1),\mu^t(S^t))$ that is $\eps$-threat-free on $S^t$ -- again, such a $\rho^t$ must exist by Proposition~\ref{prop:well-defined}.

Now, note that because $S^t$ was the last round on which $P_{-i}$ is facing an $\eps$-threat,
$P_{-i}$ is not facing an $\eps$-threat with respect to $\rho^{t-1}$ at any round following $S^t$.
Since $\rho^{t-1}$ was chosen to be $\eps$-threat free on $S^{t-1}$, player $i$ is not facing an $\eps$-threat with
respect to $\rho^{t-1}$ at any round following $S^{t-1}$. Finally, by construction, $P_i$ is not facing an $\eps$-threat at any
round following $R$ with respect to $\rho$. Since $\rho$ and $\rho^{t-1}$ are equivalent up to round $S^{t-1}$, it must be the case that
$P_i$ is not facing an $\eps$-threat with respect to $\rho^{t-1}$ at any round between $S^t$ and $S^{t-1}$ either. Thus,
 $\rho^{t-1}$ is $\eps$-threat-free on $S^t$. Since $\rho^{t-1}\in \mathrm{Cont}(\rho^{t-1}(1,\ldots,S^t))$, we then have that
 \begin{align*}
\E\left[u_{-i}\left(O(\rho^t)\right)\right] &> \E\left[u_{-i}\left(O(\rho^{t-1})\right)\right] + \eps\\
&> \E\left[u_{-i}\left(O(\rho)\right)\right] + t\cdot\eps\\
&\geq \E\left[u_{-i}\left(O(\sigma)\right)\right] + (t-1)\cdot\eps.
\end{align*}

Finally, after repeating this for all $t$ until there are no more $\eps$-threats to $P_{-i}$ at any round that follows $R$,
we are left with a profile $\rho^D\in \mathrm{Cont}(\sigma(1,\ldots,R\!-\!1),\tau(R))$ on which both $P_i$ and $P_{-i}$ are not facing an $\eps$-threat at any round that follows $R$. We can then use (\ref{eqn:dev-gain3}) to infer that
$$\E\left[u_{i}\left(O(\rho^D)\right)\right] >
\E\left[u_{i}\left(O(\sigma)\right)\right]+\eps.$$
Furthermore, $\rho^D$ satisfies
$$\E\left[u_{-i}\left(O(\rho^D)\right)\right] >\E\left[u_{-i}\left(O(\rho^{D-1})\right)\right] + D\cdot\eps \geq
\E\left[u_{-i}\left(O(\sigma)\right)\right],$$
since $D\geq 1$.

We conclude that on the profile $\rho^D$ both players strictly improve over $\sigma$, contradicting the weak
Pareto optimality of $\sigma$. Hence
no player is facing an $\eps$-threat with respect to $\sigma$ at any round $R$, and this, coupled with the fact that $\sigma$ is an
$\eps$-NE, yields that profile an $\eps$-TFNE.
\end{proof}

\fi

\section*{Acknowledgments} We thank Eddie Dekel, Oded Goldreich, Ehud Kalai, Eran Omri, and Gil Segev for helpful conversations, and the
anonymous referees for careful reading and insightful comments.

\addcontentsline{toc}{section}{References}
\bibliographystyle{plain}
\bibliography{cgt}

\ifnum\draft=2
\else
\appendix

\section{One-way Functions and Commitment Schemes}
\label{apx:defs}

A function $f$ is one-way if it is easy to compute but hard to
invert given the image of a random input. More formally,

\begin{definition}[One-way functions]\label{Definition:OWF}
A function $f : \{0,1\}^* \rightarrow \{0,1\}^*$ is said to be {\em
one-way} if the following two conditions hold:
\begin{enumerate}
\item There exists a polynomial-time algorithm that on input $x$ outputs $f(x)$.

\item For every probabilistic polynomial-time algorithm $\mathcal{A}$, every polynomial $p(\cdot)$, and all sufficiently large $n$'s
    \[ \pr{\mathcal{A}(1^n, f(U_n)) \in f^{-1} (f(U_n))} < \frac{1}{p(n)} \enspace , \]%
    where $U_n$ denotes the uniform distribution over $\{0,1\}^n$.
\end{enumerate}
\end{definition}

In this paper we also deal with one-way permutations, and we note
that the above definition naturally extends to consider
permutations.

\bigskip

A commitment scheme is a two-stage interactive protocol between a
sender and a receiver. After the first stage of the protocol, which
is referred to as the {\em commit stage}, the sender is bound to at
most one value, not yet revealed to the receiver. In the second
stage, which is referred to as the {\em reveal stage}, the sender
reveals its committed value to the receiver. For simplicity of
exposition, we will focus on bit-commitment schemes, i.e.,
commitment schemes in which the committed value is only one bit. A
bit-commitment scheme is defined via a triplet of probabilistic
polynomial-time Turing-machines $(\mathcal{S},\mathcal{R},
\mathcal{V})$ such that:
\begin{itemize}
\item $\mathcal{S}$ receives as input the security parameter $1^n$ and a bit $b$. Following its
    interaction, it outputs some information ${\sf decom}$ (the decommitment).

\item $\mathcal{R}$ receives as input the security parameter $1^n$. Following its interaction,
    it outputs a state information ${\sf com}$ (the commitment).

\item $\mathcal{V}$ (acting as the receiver in the reveal stage\footnote{Note that there is no loss of generality in assuming that the reveal stage is non-interactive. This is since any
    such interactive stage can be replaced with a non-interactive one as follows: The sender
    sends its internal state to the receiver, who then simulates the sender in the interactive
    stage.}) receives as input the security parameter $1^n$, a commitment ${\sf com}$ and a
    decommitment ${\sf decom}$. It outputs either a bit $b'$ or $\bot$.
\end{itemize}

Denote by $({\sf decom} | {\sf com}) \leftarrow \langle
\mathcal{S}(1^n , b), \mathcal{R}(1^n) \rangle$ the experiment in
which $\mathcal{S}$ and $\mathcal{R}$ interact (using the given
inputs and uniformly chosen random coins), and then $\mathcal{S}$
outputs ${\sf decom}$ while $\mathcal{R}$ outputs ${\sf com}$. It is
required that for all $n$, every bit $b$, and every pair $({\sf
decom} | {\sf com})$ that may be output by $\langle \mathcal{S}(1^n
, b), \mathcal{R}(1^n) \rangle$, it holds that $\mathcal{V} ({\sf
com}, {\sf decom}) = b$.\footnote{Although we assume perfect
completeness, it is not essential for our results.}

The security of a commitment scheme can be defined in two
complementary ways, protecting against either an all-powerful sender
or an all-powerful receiver. The former are referred to as {\em
statistically-binding} commitment schemes, whereas the latter are
referred to as {\em statistically-hiding} commitment schemes. For
simplicity, we assume that the associated ``error'' is zero,
resulting in {\em perfectly-binding} and {\em perfectly-hiding}
commitments schemes.

In order to define the security properties of such schemes, we first
introduce the following notation. Given a commitment scheme
$(\mathcal{S},\mathcal{R}, \mathcal{V})$ and a Turing machine
$\mathcal{R}^*$, we denote by ${\sf view}_{\langle
\mathcal{S}(b),\mathcal{R}^* \rangle}(1^n)$ the distribution of the
view of $\mathcal{R}^*$ when interacting with $\mathcal{S}(1^n, b)$.
This view consists of $\mathcal{R}^*$'s random coins and of the
sequence of messages it receives from $\mathcal{S}$. The
distribution is taken over the random coins of both $\mathcal{S}$
and $\mathcal{R}$. Similarly, given a Turing machine $\mathcal{S}^*$
we denote by ${\sf view}_{\langle \mathcal{S^*}(1^n),\mathcal{R}
\rangle}(1^n)$ the view of $\mathcal{S}^*$ when interacting with
$\mathcal{R}(1^n)$. Note that whenever no computational restrictions
are assumed on $\mathcal{S}^*$ or $\mathcal{R}^*$, then without loss
of generality they can be assumed to be deterministic.

\begin{definition}[Perfectly-binding commitment]\label{def:perf-hiding}
A bit-commitment scheme $(\mathcal{S}, \mathcal{R}, \mathcal{V})$ is
said to be {\em perfectly-hiding} if it satisfies the following two
properties:
\begin{itemize}
\item {\bf Computational hiding:} for every probabilistic polynomial-time Turing machine
    $\mathcal{R}^*$ the ensembles $\{ {\sf view}_{\langle \mathcal{S}(0),\mathcal{R}^*
    \rangle}(1^n) \}_{n\in \mathbb{N}}$ and $\{ {\sf view}_{\langle
    \mathcal{S}(1),\mathcal{R}^* \rangle}(1^n) \}_{n\in \mathbb{N}}$ are computationally
    indistinguishable.
\item {\bf Perfect binding:} for every Turing machine $\mathcal{S}^*$
\[ \pr{({\sf (decom, decom')} | {\sf com}) \leftarrow \langle
\mathcal{S}^*(1^n), \mathcal{R}(1^n) \rangle : \MyAtop{\mathcal{V}
({\sf com}, {\sf decom}) =
0}{\mathcal{V} ({\sf com}, {\sf decom'}) = 1 }} = 0 \enspace  , \]%
for all sufficiently large $n$, where the probability is taken over
the random coins of $\mathcal{R}$.
\end{itemize}
\end{definition}

Perfectly-binding commitments can be constructed assuming the
existence of any one-way permutation~\cite{Bl}. The construction is
``non-interactive,'' meaning that the commitment phase consists of a
single message sent from the sender $\mathcal{S}$ to the receiver
$\mathcal{R}$.

\begin{definition}[Perfectly-hiding commitment]\label{def:perf-hiding}
A bit-commitment scheme $(\mathcal{S}, \mathcal{R}, \mathcal{V})$ is
said to be {\em perfectly-hiding} if it satisfies the following two
properties:
\begin{itemize}
\item {\bf Perfect hiding:} for every Turing machine $\mathcal{R}^*$ the ensembles $\{ {\sf
    view}_{\langle \mathcal{S}(0),\mathcal{R}^* \rangle}(1^n) \}_{n\in \mathbb{N}}$ and $\{
    {\sf view}_{\langle \mathcal{S}(1),\mathcal{R}^* \rangle}(1^n) \}_{n\in \mathbb{N}}$ are
    identically distributed.
\item {\bf Computational binding:} for every probabilistic polynomial-time Turing machine
    $\mathcal{S}^*$ the exists a negligible function $\mu(n)$ so that
\[ \pr{({\sf (decom, decom')} | {\sf com}) \leftarrow \langle
\mathcal{S}^*(1^n), \mathcal{R}(1^n) \rangle : \MyAtop{\mathcal{V}
({\sf com}, {\sf decom}) = 0}{\mathcal{V} ({\sf com}, {\sf decom'})
= 1 }} < \mu(n) \enspace  , \]%
for all sufficiently large $n$, where the probability is taken over
the random coins of both $\mathcal{S}^*$ and $\mathcal{R}$.
\end{itemize}
\end{definition}

Perfectly-hiding commitments can be constructed assuming the
existence of any one-way permutation~\cite{NOVY}. This construction
is ``highly-interactive,'' in that the commitment phase requires the
exchange of $n-1$ messages between the sender and the receiver,
where $n$ is the security parameter. By relaxing the hiding
condition to be only ``statistical'' it is possible to weaken the
underlying assumption to the existence of one-way
functions~\cite{HR}. Assuming the existence of collision resistant
hash functions, it is possible to construct two-message
statistically-hiding commitments~\cite{NY,DPP}.

\fi

\end{document}